\documentclass[10pt,journal,twocolumn]{IEEEtran}
\newif\ifCLASSOPTIONromanappendices \CLASSOPTIONromanappendicestrue
\usepackage{tikz}
\usepackage[siunitx]{circuitikz}
\usepackage[utf8]{inputenc}
\usepackage{fontenc}
\usepackage{hyperref}
\usepackage{a0size}         
\usepackage{amssymb}
\usepackage{amsmath}
\usepackage{xfrac}
\usepackage{multicol}
\usepackage[english]{babel}    
\usepackage{epsfig} 
\usepackage{bm}
\usepackage{oplotsymbl}
\usepackage{bbm}
\usepackage{amsfonts,color,amsthm,amsmath, lscape,graphics}
\usepackage{enumitem}
\usepackage{url}
\usepackage{comment}
\usepackage{tabularx}
\usepackage{algorithm} 
\usepackage{algpseudocode} 
\newtheorem{theorem}{Theorem}    
\newtheorem{corollary}{Corollary}[theorem]%%%%%%%%%%%%%
\newtheorem{remark}{Remark}
\newtheorem{lemma}{Lemma}

\usepackage[font=footnotesize]{caption}
\usepackage[font=footnotesize]{subcaption}
\usepackage{subcaption}
\usepackage{epstopdf}
\usepackage{algorithm} 
\usepackage{algpseudocode}
\usepackage{cite}
\usepackage{subfiles}
\usepackage[colorinlistoftodos]{todonotes}
\definecolor{awesome}{rgb}{1.0, 0.13, 0.32}
\usepackage{fix2col}

\addto\captionsenglish{}

\usepackage{graphicx}  
\theoremstyle{plain}

\usepackage{multirow} 
\usepackage{relsize}  
\usepackage{tabularx}
\usepackage{bm}   
\usepackage{pifont}
\usepackage{etoolbox}
\usetikzlibrary{math}
\tikzmath{\trans = 5;}

\makeatletter
\patchcmd{\@maketitle}
  {\addvspace{0.5\baselineskip}\egroup}
  {\addvspace{-2.0\baselineskip}\egroup}
  {}
  {}
\makeatother

\begin{document}
\title{Age-Limited Capacity of Massive MIMO}

\author{Bamelak Tadele, Volodymyr Shyianov, Faouzi~Bellili, \IEEEmembership{Member, IEEE}, Amine Mezghani, \IEEEmembership{Member, IEEE}, and Ekram Hossain, \IEEEmembership{Fellow, IEEE}
 %\vspace{0.1cm}
%\\\small E2-390 E.I.T.C,  75 Chancellor's Circle  Winnipeg, MB, Canada, R3T 5V6.
 % \vspace{0.1cm}
 % \\\small Emails:  tadeleb@myumanitoba.ca, Faouzi.Bellili@umanitoba.ca, and Ekram.Hossain@umanitoba.ca.
  %\vspace{0.3cm}
\thanks{The authors are with the Department of Electrical and Computer Engineering at the University of Manitoba, Winnipeg, MB, Canada. (emails:\{tadeleb,shyianov\}@myumanitoba.ca, \{Faouzi.Bellili,Amine.Mezghani,Ekram.Hossain\}@umanitoba.ca). This work was supported by the Discovery Grants Program of the Natural Sciences and Engineering Research Council of Canada (NSERC).}}

\maketitle
%%%%%%%%%%%%%%%%%%%%%%%%%%%%%%%%%%%%
%% Abstract
%%%%%%%%%%%%%%%%%%%%%%%%%%%%%%%%%%%%
%\vspace{-1cm}
\begin{abstract}
We investigate the age-limited capacity of the Gaussian many channel with total $N$ users, out of which a random subset of $K_{a}$ users are active in any transmission period, and a large-scale antenna array at the base station (BS). In an uplink scenario where the transmission power is fixed among the users, we consider the setting in which both the number of users, $N$, and the number of antennas at the BS, $M$, are allowed to grow large at a fixed ratio $\zeta = {M}/{N}$. Assuming perfect channel state information (CSI) at the receiver, we derive the achievability bound under maximal ratio combining. As the number of active users, $K_{a}$, increases, the achievable spectral efficiency is found to increase monotonically to a limit $\log_2\left(1+\frac{M}{K_{a}}\right)$. Further extensions of the analysis to the zero-forcing receiver as well as imperfect CSI are provided, demonstrating the channel estimation penalty in terms of the mean squared error in estimation. Using the age of information (AoI) metric, first coined in \cite{kaul2011minimizing}, as our measure of data timeliness or freshness, we investigate the trade-offs between the AoI and spectral efficiency in the context massive connectivity with large-scale receiving antenna arrays. As an extension of \cite{liu2018massiveII}, based on our large system analysis, we provide an accurate characterization of the asymptotic (finite system size) spectral efficiency as a function of the number of antennas and the number of users, the attempt probability, and the AoI. It is found that while the spectral efficiency can be made large, the penalty is an increase in the minimum AoI obtainable. The proposed achievability bound is further compared against recent massive MIMO-based massive unsourced random access (URA) schemes.
\end{abstract}


%%%%%%%%%%%%%%%%%%%%%%%%%%%%%%%%%%%%
%% Introduction  
%%%%%%%%%%%%%%%%%%%%%%%%%%%%%%%%%%%%
\section{Introduction}
\subsection{Background and Motivation}
\IEEEPARstart{M}{assive} random access in which a base station equipped 
with a large number of antennas is serving a large number of contending users has recently attracted considerable attention. This surge of interest is fuelled by the need to satisfy the soaring demand in wireless connectivity for many envisioned IoT applications such as massive machine-type communication (mMTC). Machine-type communication (MTC) has two distinct features \cite{dutkiewicz2017massive} that make them drastically different from human-type communications (HTC) around which previous cellular systems have mainly evolved: $i)$ machine-type devices (MTDs) require sporadic access to the network and $ii)$ MTDs usually transmit small data payloads using short-packet signaling. The sporadic access leads to the overall mMTC traffic being generated by an unknown and random subset of active MTDs (at any given transmission instant or frame). This calls for the development of scalable random  access protocols that are able to accommodate a massive number of MTDs. Short-packet transmissions, however, make the traditional grant-based access (with the associated scheduling overhead) fall short in terms of spectrum efficiency and latency, which are two key performance metrics in next-generation wireless networks. Hence, a number of grant-free random access schemes have been recently investigated within the specific context of massive connectivity (see \cite{yuan2016non,Wei_Yu_paper} and references therein). From the information-theoretic point of view, the problem of massive random access is not recent and dates back to the seminal work of Gallager in  \cite{gallager1985perspective}. However, with an increasing number of possible applications the problem has reappeared in a new context \cite{chen2017capacity,polyanskiy2017perspective}. As opposed to classical treatments of the Gaussian multiple access channel in which the number of users stays fixed, in the new Gaussian many channel formalism the number of users is allowed to grow with the blocklength \cite{chen2017capacity} in a typical massive connectivity setup. Note that when a randomly varying subset of users (with different codebooks) are active over each transmission period, one is bound to sacrifice some of the spectral efficiency for user-identification \cite{chen2017capacity,liu2018massive,liu2018massiveII}. However, when all the devices employ the same codebook (aka, unsourced access), the user-identification problem can be separated from the decoding problem as highlighted in \cite{polyanskiy2017perspective}. In fact, by letting all the devices employ the same codebook, the system spectral efficiency depends on the number of active users only and not on total number of users, thereby making different multi-user decoders comparable against each other and to the random coding bound. In particular, it was shown in \cite{polyanskiy2017perspective} that increasing the number of active users at a fixed per-user payload renders known solutions such as ALOHA far from the random coding achievability bound. The paradigm in \cite{polyanskiy2017perspective} where all users share the same codebook with no need for user identification was later dubbed unsourced random access and now has a number of viable algorithmic solutions. However, most of the existing information-theoretic works on massive connectivity focus on the case of a single receive antenna at the BS. Yet, the idea of using a large-scale antenna array at the BS (i.e. massive MIMO) which was first pioneered in \cite{marzetta2010noncooperative}, has now become one of the main directions towards which the next-generation of cellular systems are projected to evolve.      \\\
\indent
From another perspective, in many real-time applications wherein the data is subject to abrupt variations, usefulness of the information when it arrives at the BS is directly related to its freshness. Due to infrequent access to the network, conventional performance metrics, such as delay fall short in characterizing the over-all freshness of the data \cite{kaul2012real}. In this respect, the AoI concept \cite{kaul2011minimizing} was introduced to adequately characterize the freshness of the information at the receiver side. While many of the existing works on AoI focus primarily on grant-based access with AoI-constrained scheduling policies \cite{jiang2019timely,kadota2019scheduling}, some have looked at uncoordinated transmission schemes. Recently, a few information-theoretic works have investigated the trade-off between the AoI and achievable data rates \cite{bastopcu2020partial,baknina2018sening}. The performance of AoI has been investigated in Multiple-Input Multiple-Output (MIMO) systems \cite{chen2020multiuser,zhu2020status,feng2022precoding,yu2021age}. In \cite{chen2020multiuser}, the user scheduling problem has been investigated to minimize AoI in a multiuser MIMO status update system where multiple single-antenna devices send their information over a common wireless uplink channel to a multiple-antenna access point. In \cite{feng2022precoding}, a novel MIMO broadcast setting is studied to minimize the sum average AoI through precoding and transmission scheduling.  In \cite{zhu2020status}, the authors analyzed and optimized the performance of AoI in a grant-free random-access system with massive MIMO.

\subsection{Contributions}
%Motivated by the aforementioned reasons, in this paper, we investigate the uplink SINR outage performance of Rayleigh-fading channels in the asymptotic regime when both the number of BS antennas and the number of users are allowed to grow large, at a fixed ratio. Assuming perfect CSI at the receiver, we derive the achievability bound under maximal ratio combining and we use it to gauge the performance of recent massive MIMO URA schemes. As the number of active users increases, the achievable spectral efficiency is found to increase monotonically up to a limit $\log_2\left(1+\frac{M}{K_{a}}\right)$, where $M$ is the number of antennas at the BS and $K_{a}$ is the number of active users. Using the concept of AoI, it is shown that fully uncoordinated non-orthogonal access can achieve minimum AoI as long as all the devices are active in each transmission period. In fact, our analysis reveals that with a large-scale antenna array at the BS both high spectral efficiency and low AoI can be achieved. 

The major contributions of this paper are summarized as follows.
\begin{itemize}
    \item We derive a closed-form expression of the outage probability in the finite-user, finite-antenna regime and through use of the central-limit theorem (CLT) we express this outage probability in the asymptotic case where both the number of users and the number of antennas are allowed to grow large at a fixed ratio.
    \item Under the assumption of perfect CSI at the receiver, we derive an achievability bound using a maximal ratio combining (MRC) receiver. We demonstrate how this achievable bound scales with the number of users in the finite regime (e.g. in Theorem \ref{supremum theorem} in Section IV) and further elaborate on its behaviour in the limit (e.g. through Theorems \ref{Age Capacity} and \ref{Strong Converse} in Section IV). 
    \item We show that fully uncoordinated non-orthogonal access can achieve minimum AoI as long as all the devices are active in each transmission period. Furthermore, our analysis reveals that with a large-scale antenna array at the BS both high spectral efficiency and low AoI can be achieved.
    \item We further extend the analysis to the case of imperfect CSI as well as the zero-forcing receiver. We derive the asymptotic as well as limiting spectral efficiency of both the MRC as well as the zero-forcing receiver when the estimation error is added to the noise contribution.
    \item Finally, using our bound, we gauge the performance of recent massive MIMO unsourced random access (URA) schemes.
\end{itemize}
The work that is most closely related to the results presented in this paper is reported in \cite{liu2018massiveII}, where the authors considered a massive connectivity with massive MIMO system for uplink data communication. In their paper, they used the state evolution framework to obtain the limiting MSE of the approximate message passing (AMP) channel estimation/activity detection algorithm. They further calculated the achievable rate (interference limited capacity) with the MRC as well as the LMMSE receiver. The limitation of their approach is that it treats only the asymptotic convergence as both the number of antennas and users are infinite while the ratio of the number of antennas to the number of users stays finite. On the other hand, the outage probability formulation together with the approximation analysis presented in this work allows for the asymptotic spectral efficiency characterization of large, yet finite, systems. The non-asymptotic point of view, Theorem \ref{supremum theorem} in the manuscript, provides the spectral efficiency as well as a precise, $\mathcal{O}(N^{-1.5})$, correction term.
\subsection{Organization of the Paper and Notations}
We structure the rest of this paper as follows. In Section \ref{section_2}, we introduce the system model. In Section \ref{section_3}, we derive the exact packet probability of error and also find its more insightful asymptotic approximation. In Section \ref{section:spectral_efficiency}, we state our main results on the trade-off between achievable spectral efficiency and the AoI. In Section \ref{section 5}, we extend the analysis to the case of imperfect CSI for the MRC and ZF receivers. These results are further corroborated by computer simulations in Section \ref{section_6}. Finally, we draw out some concluding remarks in Section \ref{Section Conclusion} and prove our various claims in the Appendices. 
\newline 
\indent We also mention the common notations used in this paper. Lower- and upper-case bold fonts, $\mathbf{x}$ and $\mathbf{X}$, are used to denote vectors and matrices, respectively. $\mathbf{I}_{M}$ denotes the $M\times M$ identity matrix. The symbols $|.|$ and $\|.\|_2$ stand for the modulus and Euclidean norm, respectively. $\{.\}^\textsf{H}$ stands for the Hermitian (transpose conjugate) operator. The shorthand notation $\mathbf{y}$ $ \sim \mathcal{CN}( \mathbf{m}, \mathbf{R})$ means that the random vector $\mathbf{y}$ follows a complex circular Gaussian distribution with mean $\mathbf{m}$ and auto-covariance matrix $\mathbf{R}$. Likewise, $S \sim \Gamma(k, \theta)$ means that the random variable $S$ follows a Gamma distribution with shape parameter $k$ and scale parameter $\theta$. The statistical expectation is denoted as $\mathbb{E}\{.\}$, and the notation $\triangleq$ is used for definitions.
%%%%%%%%%%%%%%%%%%%%%%%%%%%%%%%%%%%%%%%%%%%%%%%%%%%%%%%%%%%%%%%%%%%%%%%%%%%%%%%%%%%%%%%%%%%%%%%%%

\section{System Model, Assumptions, and Methodology of Analysis}\label{section_2}

\subsection{System Model, Assumptions, and Definition of AoI}
Consider a single-cell network consisting of $N$ single-antenna devices transmitting their status packets over an unreliable multiple-access channel to a BS with $M$ receive antenna elements. To aid synchronization, time is partitioned into slots of equal length $T$, which is the maximum amount of time for transmission and reception of a single information packet. This paper assumes sporadic device activity where at the start of every time slot user $i$ transmits is current status with probability $\tau_{i}$. We define $\{\varepsilon_i\}_{i=1}^{N}$ as the binary activity random variables which indicate whether user $i$ transmits its packet or remains idle in a given slot:
\begin{eqnarray}
\varepsilon_{i}&=& \Bigg\{ \begin{array}{ll}
{1} & \text {if user $i$}  \text { transmits his packet,~} \\
{0} & {\text {if user $i$ }  \text {remains idle.}} 
\end{array} 
\end{eqnarray}
Moreover, we assume $\{\varepsilon_i\}_{i = 1}^{N}$ are independent in each time slot with  marginal distributions $\Pr(\varepsilon_i = 1) = \tau_i$. 
%The system model of this random access set-up is shown in Fig.\ref{fig:Model_Space}.
%\begin{figure}[t]
 %   \centering
 %   \includegraphics[width=0.6\linewidth]{Picture1.jpg} 
 %   \vskip 0.1cm
 %   \caption{System model}  
 %   \label{fig:Model_Space}
%\end{figure}
To maintain timely status updates, in every slot a new packet is generated by each user. In this static macro-cell environment, where the coherence time is on the order of hundreds of milliseconds and delay spread is on the order of microseconds \cite{heath2018foundations}, the channel remains fairly constant and thus we assume a quasi-static Rayleigh fading model\textcolor{blue}{\footnote{The assumption of uncorrelated channels between the different antenna elements requires sufficiently large inter-element spacing while the assumption of uncorrelated channel vectors between users is only valid at large separation.}} for the duration of the slot in which $\mathbf{h}_i\sim\mathcal{{CN}}(\mathbf{0},\mathbf{I}_{M})$ denotes the $M\times1$ channel vector between the $i$'th user and the BS. The received signal at the BS at discrete time $n$ can then be written as:
\begin{equation}\label{received signal}
    \mathbf{y}_n ~=~ \sum_{i=1}^{N}\mathbf{h}_i\varepsilon_{i}x_{i,n}~+~\mathbf{w}_n,
\end{equation}
where $x_{i,n}\sim\mathcal{{CN}}(0,P_i)$ is the transmitted symbol, while $\mathbf{w}_n\sim\mathcal{CN}(\mathbf{{0}}_{M},\sigma_{\mathbf{w}}^2\mathbf{{I}}_{M})$ is the additive white Gaussian noise (AWGN) which is assumed spatially uncorrelated across all receive antennas. In the presence of $K_{a}$ active users in a given transmission slot, the above formulation is a $K_{a}$-user single-input multiple-output (SIMO) fading Gaussian multiple access channel (GMAC). 

Using $p_{e,i}$ to denote the slot-wise packet error probability (PEP) of the $i$th user, the probability that the $i$th user updates the BS with its current status is then given by:
\begin{equation}\label{Success Probability}
    \gamma_{i} ~=~ \tau_{i}(1 -  p_{e,i}).
\end{equation}
An example of the slotted system with $N=3$ users with $T = 1$ is depicted in Fig. \ref{fig:Model_Time}.
\begin{figure}[t]
    \centering
    \includegraphics[width=1\linewidth]{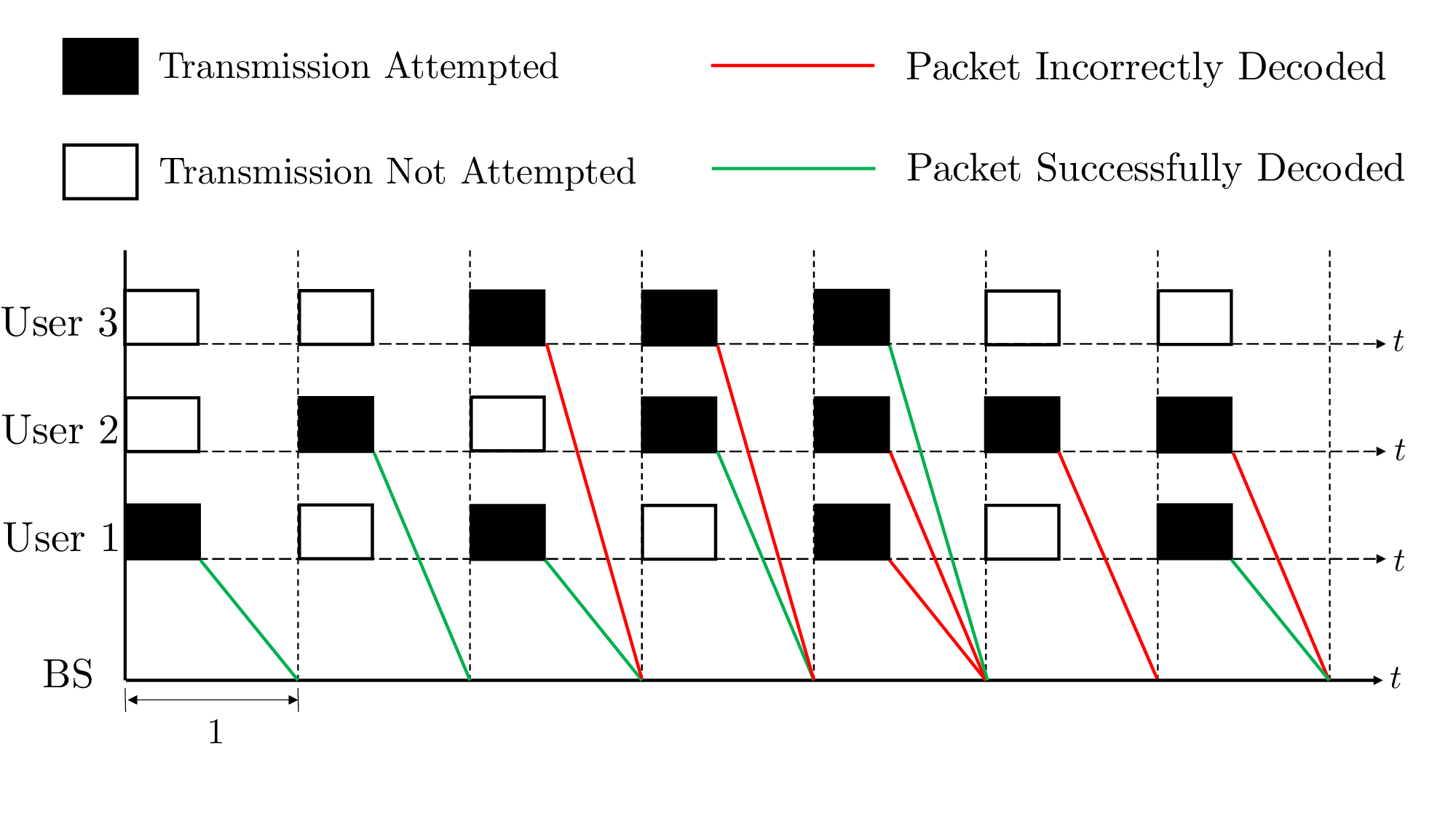}
    \caption{Slotted system with $N$ = 3 total users and unit slot length.}  
    \label{fig:Model_Time}
\end{figure}
% \begin{figure}[t]
%     \centering
%     \includegraphics[width=0.7\linewidth]{Slotted_System_diagram_v2.eps} 
%     \vskip 0.1cm
%     \caption{Slotted system with $N$ = 3 total users and unit slot length.}  
%     \label{fig:Model_Time}
% \end{figure}
Under perfect CSI at the receiver  we use maximal-ratio combining (MRC) and assume that the blocklength is sufficiently large such that the capacity limit is approached within a packet, as justified in \cite{yang2014quasi}. Consequently, $p_{e,i}$ can be closely approximated by the following outage probability: 
\begin{eqnarray} \label{Outage Probability Definition}
    \!\!1 - p_{e,i} = &\nonumber\\ & \!\!\!\!\!\!\!\!\!\!\!\!\!\!\!\!\!\!\!\!\!\Pr\left\{ \rho_i \!<\! \log_2\left(\!1+\frac{\|\mathbf{h}_i\|_2^4 P_i}{\|\mathbf{h}_i\|_2^2 \sigma_\mathbf{w}^2 + \sum_{\substack{j=1 \\j\neq i}}^{N}|{\mathbf{h}_i}^{\textsf{H}} \mathbf{h}_j|^2 \varepsilon_j P_j }\right)\right\},
\end{eqnarray}
% \begin{equation} \label{Outage Probability Definition}
%     p_{e,i} ~ =~ 1 - \Pr\left\{ \rho_i ~<~ \log_2\left(1+\frac{\|\mathbf{h}_i\|_2^4 P_i}{\|\mathbf{h}_i\|_2^2 \sigma_\mathbf{w}^2 ~+~ \sum_{\substack{j=1 \\j\neq i}}^{N}|{\mathbf{h}_i}^{\textsf{H}} \mathbf{h}_j|^2 \varepsilon_j P_j }\right)\right\},
% \end{equation}
in which $\rho_i$ [bits/channel-use] is the spectral efficiency of the $i$th user. Recall that the ultimate goal of each node is to keep the BS updated with its most recent state. If the BS has node $i$'s state that was current at time $t_0$, the age of that user's state is defined by the random process $\delta_i(t)$ $\triangleq$ $t - t_0$. When node $i$ attempts transmission and gets correctly decoded at the BS we call it an arrival. We denote $i$th node's $j$th arrival epoch by $t_{i,j}$. Between two arrival epochs, the age grows as a stair-case function of time and is reset to $T$ when an arrival occurs, since this is the amount of time that it took for a packet to be transmitted. We denote by $Z_{i,j}$ the inter-arrival time between the $j$th update and the $(j+1)$th update (i.e. $Z_{i,j}\triangleq t_{i,j+1} - t_{i,j}$). Assuming that the total number of users remains constant in each time slot, user $i$ has a certain success probability, $\gamma_i$, that it will update the BS with its state. After normalizing the slotted period to $T=1$, it then follows that each $Z_{i,j}$ is a geometric random variable with parameter $\gamma_i$. 

Similar to \cite{yates2017status}, we define the AoI, $\Delta_{i}$, of each node $i$ as: 
\begin{equation}\label{Age of Information}
    \Delta_{i} ~\triangleq~ \lim_{t'\to \infty}\frac{1}{t'}\int_{0}^{t'}{\delta_i(t)dt}.
\end{equation}
 {For completeness, we show in \textbf{Appendix \ref{appendix:Age}} that the limit in (\ref{Age of Information}) exists and that it converges with probability one (WP1) to
\begin{equation}\label{Limit of Age of User i}
    \Delta_i ~=~ \frac{\mathbb{E}[Z_i^2]}{2\mathbb{E}[Z_i]} ~+~ \frac{1}{2},
\end{equation}
as was done in more general terms in \cite{li2014throughput}.}
Since the $Z_{i}$'s are geometric random variables, it follows that $\mathbb{E}[Z_{i}] ~=~ \frac{1}{\gamma_{i}}$ and $\mathbb{E}[Z_{i}^2] ~=~ \frac{2}{\gamma_{i}^2} - \frac{1}{\gamma_{i}}$, thereby leading to\footnote{Notice that since the error probability of the $i$th user $p_{e,i}$ is averaged over the number of active users, the arrival process is still Bernoulli}:
\begin{equation}\label{Age of information of Users i}
    \Delta_i ~=~ \gamma_{i}^{-1}.
\end{equation}

In the presence of $N$ total users, we consider the network-wide average AoI, given by:  
\begin{equation}\label{definition of network AoI}
    \Delta = \frac{1}{N}\sum_{i=1}^{N} \Delta_i.
\end{equation}

\subsection{Methodology of Analysis}

Since the AoI solely depends on the parameter $\tau_i$ and the outage probability $p_{e,i}$, the main objective will be to find the outage probability. In the following, we will first do this in closed-form and thereafter, through use of the central limit theorem, we will find an approximation in the asymptotic regime. This later allows us to find an explicit relationship between the spectral efficiency, $\tau$, the ratio $M/N$, and the probability of error. To this end, we will illustrate exactly how the spectral efficiency scales in the finite-user, finite-antenna case (e.g., through \textbf{Theorem \ref{supremum theorem}} in Section IV). Thereafter, we will take the limit as the number of users and the number of antennas grow large and we will find a phase-transition where the AoI is minimized in one regime and grows unbounded in the other (as will be stated in \textbf{Theorem \ref{Age Capacity}} and \textbf{Theorem \ref{Strong Converse}} in Section IV).

%%%%%%%%%%%%%%%%%%%%%%%%%%%%%%%%%%%%%%%%%%%%%%%%%%%%%%%%%%%%%%%%%%%%%%%%%%%%%%%%%%%%%%%%%%%%%%%%%

%In the static scenario described previously, if the time slots are a fraction of the coherence time then each user can estimate the channel at the start of a fading block and use this estimate for the remainder of the coherence block whenever they transmit. However, in a sporadic access environment where users will only access the channel after long periods of time this assumption falls apart and analysis with pilot overheads will be necessary.%

\section{Derivation of Packet Error Probability}\label{section_3}
\subsection{Derivation of Exact PEP}
Recall from (\ref{Success Probability}) that in order for user $i$ to successfully update the BS with its status in a given slot $a)$ it must attempt a transmission in that slot and $b)$ the transmitted packet must be decoded correctly. For ease of analysis, we consider a symmetric system wherein the $N$ users have the same transmit power (i.e. $P_i=P$ $\forall i$) and the same attempt probability (i.e. $\tau_i = \tau$ $\forall i$). Dividing the second term inside the logarithm in (\ref{Outage Probability Definition}) by $\|\mathbf{h}_i\|_2^2$ and rearranging the terms we obtain:
\begin{equation}\label{simplified_pe_first_time}
   \!\!\!1-p_{e,i} = \Pr\left\{ \left(2^\rho -1\right)\Bigg(\sigma_\mathbf{w}^2\!+\!\sum_{\substack{j=1 \\j\neq i}}^{N}|P \varepsilon_j \widetilde{\mathbf{h}}_i^{\textsf{H}} \mathbf{h}_j|^2 \Bigg) \!\leq\! \|\mathbf{h}_i\|_2^2 P\right\}
\end{equation}
where $\widetilde{\mathbf{h}}_i = \frac{\mathbf{h}_i}{\|\mathbf{h}_i\|_2}$.  {For notational compactness we define:
\begin{equation}
    \alpha_\rho = \frac{1}{2^\rho - 1},
\end{equation}
and we denote the inverse signal-to-noise ratio (SNR) as $\beta~\triangleq~\frac{\sigma_\mathbf{w}^2}{P}$.} Using these notations and further simplifying (\ref{simplified_pe_first_time}), we obtain:
\begin{equation}\label{simplified pi}
   1 - p_{e,i} ~=~ \Pr\left\{ \alpha_\rho \|\mathbf{h}_i\|_2^2 ~-~ \sum_{\substack{j=1 \\j\neq i}}^{N}|\widetilde{\mathbf{h}}_i^{\textsf{H}} \mathbf{h}_j|^2 \varepsilon_j ~\geq~ \beta \right\}.
\end{equation}

As was shown in \cite{shah2000performance}, $\widetilde{\mathbf{h}}_i^{\textsf{H}}\mathbf{h}_j\sim\mathcal{CN}(0,1)$ $\forall$ $i,j$ and they are mutually independent and also independent of $\|\mathbf{h}_i\|_2^2$. In order to deal with the random sum in (\ref{simplified pi}), we condition on the event of having $k$ other users being active with the $i$th user. Since the $\varepsilon_j$'s are independent and identically distributed (i.i.d) Bernoulli random variables (RVs), the probability that $k$ users are active out of the remaining $N-1$ users (i.e. after excluding user $i$) is the same as having $k$ successes in $N-1$ Bernoulli trials. Thus, the number of active users follows a Binomial distribution with parameters $N-1$ and $\tau$. In order to calculate $p_{e,i}$ for each $i$th user it is convenient to marginalize over the number of other active users, thereby leading to:
\begin{equation}\label{Total Average}
    p_{e,i} ~=~ 1 ~-~ \sum_{k=0}^{N-1}{{N-1} \choose {k}} {\tau^{k}}{(1-\tau)^{N-1-k}} p_{i | k},
\end{equation}
where $p_{i|k}$ is defined as the conditional probability of successful decoding, conditioned on $k$ other users being also active. More specifically, we have:
\begin{equation}\label{conditional probability}
    p_{i|k} ~=~ \Pr\left\{ \alpha_\rho \|\mathbf{h}_i\|_2^2 ~-~ \sum_{j=1}^{k}|\widetilde{\mathbf{h}}_i^{\textsf{H}} \mathbf{h}_j|^2  ~\geq~ \beta \right\},
\end{equation}
where $\|\mathbf{h}_i\|_2^2$ follows a gamma distribution with shape parameter $M$ and scale parameter $1$ (i.e., $||\mathbf{h}_i||_{2}^2 \sim \Gamma(M,1)$). Similarly, the second term in (\ref{conditional probability}) is a sum of $k$ complex normal RVs squared and hence follows a $\Gamma(k,1)$ distribution. By defining $H$~$\triangleq$~$\alpha_\rho \|\mathbf{h}_i\|_2^2$ and $X_k$~$\triangleq$~$\sum_{j=1}^{k}|\widetilde{\mathbf{h}}_i^{\textsf{H}} \mathbf{h}_j|^2$, we see that $p_{i|k}$ in (\ref{conditional probability}) is the complementary distribution function of the RV $Z$ $\triangleq$ $H - X_{k}$, as a function of the inverse SNR. In the case there are no other active users (i.e. $k=0$), $p_{i|k}$ is given by the complementary distribution of a gamma RV\footnote{A gamma RV, with scale parameter $\theta$, multiplied by a real number $\alpha_{\rho}$, is another gamma RV with scale parameter $\alpha_{\rho}\theta$.}. For $k>0$, however, one can find the probability density function (pdf) of $Z$ through convolution, thereby leading to:
\begin{equation}\label{pdf of Z}
\!\!\!\!\!\!f_{Z}(z)=\left\{\begin{array}{ll}
{\kappa\mathlarger{\int_{-\infty}^{~z}(-x)^{k-1}(z-x)^{M-1}e^{2^{\rho}x}dx}} & {\text {\!\!\! if $z < 0$} } \\
{\kappa\mathlarger{\int_{-\infty}^{~0}(-x)^{k-1}(z-x)^{M-1}e^{2^{\rho}x}dx}} & {\text {\!\!\!if $z$ } \geq \text {0,}}
\end{array}\right.
\end{equation}
where $\kappa = \frac{e^{\frac{-z}{\alpha_\rho}}}{(k-1)!(M-1)!\alpha_{\rho}^M}$.
% \begin{equation}\label{pdf of Z}
% f_{Z}(z)~=~\left\{\begin{array}{ll}
% {\frac{e^{\frac{-z}{\alpha_\rho}}}{(k-1)!(M-1)!\alpha_{\rho}^M}\mathlarger{\int_{-\infty}^{~z}(-x)^{k-1}(z-x)^{M-1}e^{2^{\rho}x}dx}} & {\text { if $z < 0$} } \\
% {\frac{e^{\frac{-z}{\alpha_\rho}}}{(k-1)!(M-1)!\alpha_{\rho}^M}\mathlarger{\int_{-\infty}^{~0}(-x)^{k-1}(z-x)^{M-1}e^{2^{\rho}x}dx}} & {\text {if $z$ } \geq \text {0.}}
% \end{array}\right.
Since we are primarily interested in the probability that $Z$ is greater than\footnote{Recall here that $\beta$ is the inverse SNR which is a positive quantity.} $\beta>0$, we are only concerned with $f_{Z}(z)$ for non-negative values of $z$. By further manipulating the integral (\ref{pdf of Z}), it can be shown that the pdf for $z\geq0$ can be written as:
\begin{eqnarray}\label{Whittaker function}
    \!\!\!f_Z(z) = &\nonumber\\& \!\!\!\!\!\!\!\!\!\!\!\!\!\!\!\!\!\!\!\!\!\!\!\!\!\!\!{\frac{z^{\left(\frac{M+k-2}{2}\right)}}{(M-1)!\alpha_{\rho}^{M}  2^{\frac{\rho}{2} (M+k)}}}{\exp{\left(-\frac{z}{2}(2^{\rho} - 2)\right)}}{\widetilde{W}_{\frac{M-k}{2},\frac{1-M-k}{2}}(2^{\rho}z)}
\end{eqnarray}
% \begin{equation}\label{Whittaker function}
%     f_Z(z) ~=~ {\frac{z^{\left(\frac{M+k-2}{2}\right)}}{(M-1)!\alpha_{\rho}^{M}  2^{\frac{\rho}{2} (M+k)}}}{\exp{\left(-\frac{z}{2}(2^{\rho} - 2)\right)}}{\widetilde{W}_{\frac{M-k}{2},\frac{1-M-k}{2}}(2^{\rho}z)},
% \end{equation}
where $\widetilde{W}(.)$ denotes the Whittaker function. Averaging $p_{i|k}$ over the number of active users and incorporating everything together we can finally write $p_{e,i}$ as follows:
\begin{eqnarray}\label{Non-asymptotic Pi}
    p_{e,i} = 1 -  (1-\tau)^{N-1}\big(1-\Pr\{H\leq\beta\}\big) &\nonumber\\& \!\!\!\!\!\!\!\!\!\!\!\!\!\!\!\!\!\!\!\!\!\!\!\!\!\!\!\!\!\!\!\!\!\!\!\!\!\!\!\!\!\!\!\!\!\!\!\!\!\!\!\!\!\!\!\!\!\!\!\!\!\!\!\!\!\!\!\!\!\!\!\!\!\!\!\!\!\!\!\!\!\!\!\!\!\!\!\!\!\!- \sum_{k=1}^{N-1}{N-1\choose k}\tau^{k}(1-\tau)^{N-1-k}\int_{\beta}^{\infty}f_{Z}(z)dz.
\end{eqnarray}

\subsection{Asymptotic Approximation of PEP}\label{section:assymptotic PEP}
While (\ref{Non-asymptotic Pi}) is an exact expression for PEP, it does not provide insights into the scaling law of error probability as the total number of users and the number of BS antenna branches both increase at a fixed ratio. In this Section, we derive an asymptotic approximation of PEP which becomes increasingly exact in the large system limit. More specifically, we will let the number of users $N$ and the number of antennas $M$ grow large, while keeping their ratio, $\zeta\,\triangleq\,\frac{M}{N}$, constant. 
%In (\ref{Total Average}) and (\ref{conditional probability}) we see  two separate sums of random variables which grow with $N$ and $M$. The first term in (\ref{Total Average}) is a sum of $N$ Bernoulli random variables which we can approximate using the De Moivre-Laplace approximation to the Binomial distribution. The second term is also a large sum which grows with $M$ and $N$ but the RVs involved are not identically distributed. We can apply the Lyapunov Central Limit Theorem under some conditions, detailed in the appendix. We re-write equation (\ref{Total Average}) as,%

The analysis technique utilized in what follows capitalizes on the Berry-Esseen theorem. Proofs of the various claims introduced in this Section are detailed in \textbf{Appendix \ref{Appendix:CLT}}. Using symmetry arguments, it can be seen that $p_{e,i}$ does not depend on $i$ and after omitting that index it follows from (\ref{Total Average}) that:
\begin{eqnarray}\label{Explicit Total Probability}
   1 - p_e = \sum_{k=0}^{N-1} \Pr\left\{\varepsilon_1+\ldots+\varepsilon_{N-1} = k \right\}G_{M,k}(\beta),
\end{eqnarray}
where $G_{M,k}(\beta) = \Pr\left\{\sum_{m=1}^{2M}\!\!S_{m}\!\!+\!\!\sum_{l=1}^{2k}V_{l} \geq \beta\right\}$
in which $S_m\sim \Gamma(\frac{1}{2},\alpha_{\rho})$ and $V_l\sim \Gamma(\frac{1}{2},1)$.
% \begin{equation}\label{Explicit Total Probability}
%     p_e ~=~ 1~-~\sum_{k=0}^{N-1} \Pr\left\{\varepsilon_1+\ldots+\varepsilon_{N-1} = k \right\}\Pr\left\{\sum_{m=1}^{2M}S_{m}+\sum_{l=1}^{2k}V_{l} \geq \beta\right\},
% \end{equation}
% in which $S_m\sim \Gamma(\frac{1}{2},\alpha_{\rho})$ and $V_l\sim \Gamma(\frac{1}{2},1)$.
%with shape parameter $\frac{1}{2}$ and scale parameter $\alpha_{\rho}$. First and second order statistics are obtained from the distribution as $\mathbb{E}[S_m] = \frac{\alpha_\rho}{2}$ and $\sigma_{S_{m}}^2$ = $\frac{\alpha_{\rho}^2}{2}$.Similarly,  have mean $\mathbb{E}[V_l] = -\frac{1}{2}$ and $\sigma_{V_l}^2$ = $\frac{1}{2}$.
%\begin{equation}\label{\CLT convergence}
   % \frac{1}{\sqrt{\alpha_{\rho}^{2}M~+~k}}\left[\sum_{i=1}^{2M}S_{i}+\sum_{j=1}^{2k}V_{j}~-~(M\alpha_{\rho}-k) \right]~\xrightarrow[]{}~\mathcal{N}(0,1),
%\end{equation}
Using Berry-Essen central limit theorem (BE-CLT), the inverse cumulative distribution function (CDF) of the sum of gamma RVs in (\ref{Explicit Total Probability}) converges uniformly to the standard normal inverse CDF (see Lemma \ref{Lemma_1} in \textbf{Appendix \ref{Appendix:CLT}}), i.e.
\begin{equation}\label{large_sum_control}
    \Pr\left\{\sum_{m=1}^{2M}S_{m}+\sum_{l=1}^{2k}V_{l} \geq \beta\right\} ~=~ Q(w(k))~+~\mathcal{O}\left(\frac{1}{\sqrt{M+k}}\right),
\end{equation}
where $w(k) = \frac{\beta - \alpha_{\rho}M+k}{\sqrt{\alpha_{\rho}^{2}M+k}}$ and $Q(.)$ is the standard Q-function, (i.e., the tail of the normal distribution):
\begin{eqnarray}\label{Q-function definition}
Q(x)~=~\frac{1}{\sqrt{2{\pi}}}\int_{x}^{+\infty}\mathlarger{e^{\frac{-t^2}{2}}}dt.
\end{eqnarray}
%\begin{equation}
   % S ~\triangleq~ \frac{S_{N-1}-(N-1)\tau}{\sqrt{(N-1)\tau(1-\tau)}}~\xrightarrow[]{}~\mathcal{N}(0,1)
%\end{equation}
%\begin{equation}
 %    \sum_{k=p'}^{u'} \Pr\left\{\varepsilon_1+\ldots+\varepsilon_{N-1} = k \right\} = \frac{1}{\sqrt{2{\pi}}}\int_{p}^{u}\mathlarger{e^{\frac{-s^2}{2}}}ds + \mathcal{O}\left(\frac{1}{\sqrt{N}}\right),
%\end{equation}
%where $p = \frac{p'-(N-1)\tau}{\sqrt{(N-1)\tau(1-\tau)}}$ and $u = \frac{u'-(N-1)\tau}{\sqrt{(N-1)\tau(1-\tau)}}$,
Now incorporating the result in (\ref{large_sum_control}) into (\ref{Explicit Total Probability}) and then using the CLT on $\Pr\{\varepsilon_1+\ldots+\varepsilon_{N-1} = k \}$, (\ref{Explicit Total Probability}) can be re-written as (see Lemma \ref{appendix_A:Q_function} in \textbf{Appendix \ref{Appendix:CLT}}):
\begin{equation}\label{approximation of Pi}
  p_{e} ~=~ 1 ~-~ \frac{1}{\sqrt{2 \pi}}\int_{-\infty}^{\infty}Q(w(s))e^{-\frac{s^2}{2}}ds ~+~ \mathcal{O}\left(\frac{1}{\sqrt{N}}\right),
\end{equation}
where
\begin{equation}\label{before N-1 Approximation}
    w(s) ~=~ \frac{\beta-\alpha_{\rho}M+s\sqrt{(N-1)\tau(1-\tau)}+(N-1)\tau}{\sqrt{\alpha_{\rho}^{2}M+s\sqrt{(N-1)\tau(1-\tau)}+(N-1)\tau}}.
\end{equation}
For large $N$, we approximate $N-1$ by $N$ (see Lemma \ref{lemma_3} in \textbf{Appendix \ref{Appendix:CLT}}) thereby leading to:
\begin{equation}\label{Approximating N-1}
    w(s) ~=~ \frac{\beta-\alpha_{\rho}M+s\sqrt{N\tau(1-\tau)}+N\tau}{\sqrt{\alpha_{\rho}^{2}M+s\sqrt{N\tau(1-\tau)}+N\tau}}~ +~ \mathcal{O}\left(\frac{1}{\sqrt{N}}\right).
\end{equation}
Then by substituting $\zeta = \frac{M}{N}$ and multiplying both the numerator and denominator by $\frac{1}{\sqrt{N}}$ we obtain:
\begin{equation}\label{before neglecting terms}
       w(s) ~=~ \frac{\frac{\beta}{\sqrt{N}}-\alpha_{\rho}\zeta\sqrt{N}+s\sqrt{\tau(1-\tau)}+\sqrt{N}\tau}{\sqrt{\alpha_{\rho}^{2}\zeta+\frac{s}{\sqrt{N}}\sqrt{\tau(1-\tau)}+\tau}}+ \mathcal{O}\left(\frac{1}{\sqrt{N}}\right).
\end{equation}
We further neglect the terms which vanish for large $N$ in (\ref{before neglecting terms}), thereby leading to (see Lemma \ref{lemma_4} in \textbf{Appendix \ref{Appendix:CLT}}),
\begin{equation}\label{neglecting stuff with s}
    w(s) = \frac{\sqrt{N}(\tau-\alpha_{\rho}\zeta)+s\sqrt{\tau(1-\tau)}}{\sqrt{\alpha_{\rho}^{2}\zeta+\tau}} + \mathcal{O}\left(\frac{1}{\sqrt{N}}\right).
\end{equation}
We can also neglect the second term in the numerator of (\ref{neglecting stuff with s}) that involves the integration variable $s$. In fact, although $s$ grows large inside the integral, the exponential makes the integrand function vanish for large-magnitude values of $s$. Small values of $s$, however, can also be neglected for large values of $N$ (i.e., in the asymptotic regime). Finally, our approximation for $w(s)$ makes it independent of $s$ (see Lemma \ref{lemma_5} in \textbf{Appendix \ref{Appendix:CLT}}):
\begin{equation}\label{final simplification of w(s)}
    w ~=~ \frac{\sqrt{N}(\tau-\alpha_{\rho}\zeta)}{\sqrt{\alpha_{\rho}^{2}\zeta+\tau}} + \mathcal{O}(1).
\end{equation}
Consequently, one can take $Q(w(s))$ outside of the integral in (\ref{approximation of Pi}) (see Lemma \ref{lemma_6} in \textbf{Appendix \ref{Appendix:CLT}}):
\begin{equation}\label{taking Q out}
    p_{e} ~=~ 1-  Q(w)\int_{-\infty}^{\infty}\frac{1}{\sqrt{2\pi}}e^{-\frac{s^2}{2}}ds + \mathcal{O}\left(\frac{1}{\sqrt{N}}\right),
\end{equation}
which simplifies to: 
\begin{equation}\label{Asymptotic Probability of Success}
    p_{e} ~=~ 1 - Q(w) + \mathcal{O}\left(\frac{1}{\sqrt{N}}\right).
\end{equation}

\section{ Trade-Off Between AoI and Spectral Efficiency}\label{section:spectral_efficiency}

In this Section, we characterize the trade-off between the achievable spectral efficiency and the AoI in multiuser systems with a large-scale antenna array at the BS. It is shown that as the number of users, $N$, and the number of antennas, $M$, increase while keeping their ratio constant (i.e. $\zeta = \frac{M}{N}$) the maximum achievable spectral efficiency  approaches a well-characterized limit for any fixed AoI. The trade-off is manifested by making an observation that spectral efficiencies above the established limit can only be achieved by increasing the overall system AoI. To that end, we rewrite (\ref{definition of network AoI}) more explicitly as a function of the system parameters (see Lemma \ref{lemma_7} in \textbf{Appendix \ref{Appendix:CLT}}):
\begin{equation}\label{Clean AoI}
    \Delta(\zeta,N,\rho,\tau) ~=~ \frac{1}{\tau\left(1-Q\left(\frac{\sqrt{N}(\alpha_{\rho}\zeta-\tau)}{\sqrt{\alpha_{\rho}^2 \zeta +\tau}}\right)\right)} + \mathcal{O}\left(\frac{1}{\sqrt{N}}\right),
\end{equation}
from which it follows, in the limit, that the minimum AoI for a given attempt probability $\tau$ is given by\footnote{One can also numerically optimize (\ref{Clean AoI}), ignoring the error term, and find the $\tau$ that minimizes (\ref{Clean AoI}). Ignoring the error term will have minimal affect in this optimization as seen in the real-time simulation in Fig. \ref{fig:AoI_vs_rho}.}:
\begin{equation}\label{min AoI_tau}
    \Delta_\textrm{min}(\tau) ~=~ \frac{1}{\tau}.
\end{equation}
 {We start by defining, for a given $N$, $\tau$, and $\epsilon>0$, the set:
\begin{equation}\label{Set of achievable rates}
\Psi_{\epsilon} ~\triangleq~ \{\rho\in \Re^{+} ~|~p_{e} < \epsilon\},
\end{equation}
as the set of all achievable spectral efficiencies for which the probability of error is less than $\epsilon$. Note also that the condition $p_e < \epsilon$ implies that $\Delta(\zeta, N, \rho, \tau) < \frac{1}{\tau(1-\epsilon)}+\mathcal{O}\big(\frac{1}{\sqrt{N}}\big)$. We illustrate fundamental trade-offs in finite user case in the following theorem.
\begin{theorem}\label{supremum theorem}
For any $0<\tau<1$, $\zeta >0$, and $\epsilon>0$, there exist $N_0\in \mathbb{N}$ such that for any 
$N>N_{0}$, the set $\Psi_{\epsilon}$ is non-empty with a supremum
\begin{eqnarray}\label{Supremum}
    \rho_{N}^{*} ~\triangleq~\sup_{\rho}\Psi_{\epsilon} = &\\& \!\!\!\!\!\!\!\!\!\!\!\!\!\!\!\!\!\!\!\!\!\!\!\!\!\!\!\!\!\!\!\!\!\!\!\!\!\!\!\!\!\!\!\!\! \log_2\left(\!\!{1+\frac{\zeta-\frac{Q^{-1}(\epsilon)^2}{N}}{\tau+\sqrt{\tau^2+\tau\left(1-\frac{Q^{-1}(\epsilon)^2}{N\zeta}\right)\left(\frac{Q^{-1}(\epsilon)^2}{N}-\tau\right)}}}\right) + \mathcal{O}\left(\frac{1}{N^{1.5}}\right). \nonumber
\end{eqnarray}
% \begin{eqnarray}\label{Supremum}
%     \!\!\!\!\rho_{N}^{*} \triangleq\sup_{\rho}\Psi_{\epsilon} = &\nonumber\\& \!\!\!\!\!\!\!\!\!\!\!\!\!\!\!\!\!\!\!\!\!\!\!\!\!\!\!\!\!\! \log_2\left(\!\!{1+\frac{\zeta-\frac{Q^{-1}(\epsilon)^2}{N}}{\tau+\sqrt{\tau^2+\tau\left(1-\frac{Q^{-1}(\epsilon)^2}{N\zeta}\right)\left(\frac{Q^{-1}(\epsilon)^2}{N}-\tau\right)}}}\right) &\nonumber\\& ~~~~~~~~~~~~~~~~~~~~~~~~~~~~~~~~+ \mathcal{O}\left(\frac{1}{N^{1.5}}\right).
% \end{eqnarray}
\end{theorem}
\begin{proof}
As shown in \textbf{Appendix \ref{Appedix: Theorem 3 conditions}}, the condition that $p_{e}<\epsilon$ leads to $\alpha_{\rho}>\alpha_{\rho}^{+}(\epsilon_1)$ where:
\begin{equation}\label{positive root}
    \alpha_{\rho}^{+}(\epsilon_1) ~=~ \frac{\tau+\sqrt{\tau^2 -\tau(1-\frac{Q^{-1}(\epsilon_1)^2}{N\zeta})(\tau-\frac{Q^{-1}(\epsilon_1)^2}{N})}}{\zeta-\frac{Q^{-1}(\epsilon_1)^2}{N}},
\end{equation}
and $\epsilon_1 = \epsilon+\mathcal{O}\big(1/\sqrt{N}\big)$. Due to the differentiability of $Q^{-1}$, the error term can be taken out of $Q^{-1}$ in (\ref{positive root}) thereby leading to:
\begin{eqnarray}\label{positive root 2}%we choose the (-) so pe = epsilon not 1-epsilon
    \alpha_{\rho}^{+}(\epsilon_1)=&\nonumber\\&   \!\!\!\!\!\!\!\!\!\!\!\!\!\!\!\!\!\!\!\!\!\!\!\!\!\! \frac{\tau+\sqrt{\tau^2 -\tau\left(1-\frac{Q^{-1}(\epsilon)^2}{N\zeta}+\mathcal{O}\big(\frac{1}{N^{1.5}}\big)\right)\left(\tau-\frac{Q^{-1}(\epsilon)^2}{N}+\mathcal{O}\big(\frac{1}{N^{1.5}}\big)\right)}}{\zeta-\frac{Q^{-1}(\epsilon)^2}{N}+\mathcal{O}\big(\frac{1}{N^{1.5}}\big)}.  
\end{eqnarray}
Recalling the definition of $\alpha_{\rho}$, we see that $\alpha_{\rho}>\alpha_{\rho}^{+}(\epsilon_1)$ is equivalent to:
\begin{equation}\label{Condition so Pe 0}
    \rho ~<~ \log_2{\left(1+\frac{1}{\alpha_{\rho}^{+}(\epsilon_1)}\right)}.
\end{equation}
Again, due to the differentiability of the logarithm and due to $\epsilon$ being greater than $\epsilon_{0}(N,\zeta)$ the error term can be taken out of the logarithm and we have:
\begin{equation}\label{Condition so Pe 0 re-written}
    \rho ~<~ \log_2{\left(1+\frac{1}{\alpha_{\rho}^{+}(\epsilon)}\right)}+ \mathcal{O}\left(\frac{1}{N^{1.5}}\right).
\end{equation}
Therefore, $\Psi_{\epsilon}$ can be re-written as:
\begin{equation}\label{Theorem_3_before_the_end}
    \Psi_{\epsilon} ~=~ \left\{\rho\in \Re^{+} ~\middle|~ \rho < \log_2{\left(1+\frac{1}{\alpha_{\rho}^{+}(\epsilon)}\right)+ \mathcal{O}\left(\frac{1}{N^{1.5}}\right)}\right\}.
\end{equation}
Note that the upper bound on $\rho$ is always positive as we assume $\epsilon>\epsilon_{0}(N,\zeta)$ and hence $\Psi_{\epsilon}$ is non-empty and its is supremum is given by (\ref{Supremum}).
\end{proof}}
 {A special case of \textbf{Theorem \ref{supremum theorem}} wherein the error probability vanishes in the limit is described in the following two theorems.}
\begin{theorem}\label{Age Capacity}
(Achievability) For any $0<\tau<1$ and $\zeta>0$, we define the age-limited capacity as 
\begin{equation}\label{Age_Capacity_with_tau}
    C_{\tau,\zeta} ~=~ \log_2{\left(1+\frac{\zeta}{\tau}\right)},
\end{equation}
such that for any spectral efficiency, $\rho < C_{\tau,\zeta}$, the error probability, $p_{e} \xrightarrow[]{} 0$, and the AoI, $\Delta(\zeta,N,\rho,\tau) \xrightarrow[]{} \Delta_\textrm{min}(\tau)$, as $N\xrightarrow[]{} \infty$.
\end{theorem}
\begin{proof}
Note that the second term in (\ref{Clean AoI}) goes to zero in the limit as $N\to\infty$ and so the AoI is determined by the first term. Given the parameters $\tau$ and $\zeta$, we see that the AoI in (\ref{Clean AoI}) is monotonically increasing with $p_e$. Therefore, as $p_{e} \xrightarrow[]{} 0$ the AoI $\Delta(\zeta,N,\rho,N) \xrightarrow[]{} \Delta_\textrm{min}(\tau)$. Now, fix $\delta > 0$ and $\rho \,=\, C_{\tau,\zeta} - \delta$. The probability of error is determined by the Q-function or equivalently its argument. Hence, for a given $N$ and a given value of
\begin{equation}\label{Pe definition finite}
     \varphi(\rho) ~\triangleq~  \frac{\alpha_{\rho}\zeta-\tau}{\sqrt{\alpha_{\rho}^2 \zeta +\tau}},
\end{equation}
the probability of error is well specified. 
Plugging $\rho$ in (\ref{Pe definition finite}), it follows that:
\begin{equation}\label{Theorem_1:before_the_end}
    \varphi(C_{\tau,\zeta} - \delta) ~=~ \frac{(\zeta + \tau)(1-2^{-\delta})}{\sqrt{\zeta+\tau(2^{(C_{\tau,\zeta}-\delta)}-1)^2}},
\end{equation}
which is always positive. Therefore, as $N \xrightarrow[]{} \infty$ the argument inside the Q-function approaches $+\infty$; hence $p_{e} \xrightarrow[]{} 0$ and $\Delta(\zeta, N, \rho, \tau) \xrightarrow[]{} \Delta_\textrm{min}(\tau)$.
\end{proof}

\begin{theorem}\label{Strong Converse}
Given $\tau$, $\zeta$ and any spectral efficiency $\rho> C_{\tau,\zeta}$, the error probability, $p_{e}\xrightarrow[]{}1$, and the AoI $\Delta(\zeta, N, \rho, \tau)\xrightarrow[]{}\infty$ as, $N\xrightarrow[]{}\infty$.
\end{theorem}
\begin{proof}
We prove this in a similar way as we did for \textbf{Theorem \ref{Age Capacity}}. In fact, we choose an arbitrary $\delta>0$ and set $\rho = C_{\tau,\zeta}+\delta$. Plugging the latter in (\ref{Pe definition finite}) and simplifying we obtain:
\begin{equation}
    \varphi(C_{\tau,\zeta} + \delta) ~=~ -~\frac{(\zeta + \tau)(2^{\delta}-1)}{\sqrt{\zeta+\tau(2^{(C_{\tau,\zeta}+\delta)}-1)^2}}.
\end{equation}
Now, the argument inside the Q-function is negative for all values of $N$ and approaches $-\infty$ as $N \xrightarrow[]{} \infty$, from which it follows that $p_{e} \xrightarrow[]{} 1$ and $\Delta(\zeta, N, \rho, \tau) \xrightarrow[]{}\infty$.
\end{proof}

\begin{corollary}\label{corrolary:limit_capacity}
For any $\epsilon>0$, the age-limited capacity defined in (\ref{Age_Capacity_with_tau}) is given by 
\begin{equation}
    C_{\tau,\zeta} ~=~ \lim_{N\xrightarrow[]{}\infty}{\rho_{N}^{*}}.
\end{equation}
\end{corollary}
\begin{proof}
We see that as $N\xrightarrow[]{}\infty$ the threshold, $\epsilon_{0}(N,\zeta)$ goes to zero as it decays with $N$. Therefore, for any $\epsilon>0$, $\rho_{N}^{*}$ is well defined. As the function $\log_2{x}$ is continuous at $x = 1+\frac{\zeta}{\tau}$, one can take the limit inside its argument in $\rho_{N}^*$, from which the corollary follows.
\end{proof}

\begin{remark}\label{Shannon Capacity comparison}
It is interesting to observe the similarity between the age-limited capacity and the capacity of the AWGN channel. In the age-limited capacity, the ratio $\frac{\zeta}{\tau}$ plays the role of the SNR in the AWGN Capacity. In working with asymptotic scenarios where we have both a large number of users and a large number of antennas, the noise variance becomes negligible. In this asymptotic interference-limited scenario the decoding error probability is dominated by $\tau$. It is insightful in this case to view $\tau$ as the noise variance. Similarly, the ratio, $\zeta$, of the number of antennas to the number of users plays the role of the transmit power.
\end{remark}
\begin{remark}
Note also, that the age-limited capacity, $C_{\tau,\zeta} = \log_2\left(1+\frac{\zeta}{\tau}\right)$, is parameterized by $\tau$ and $\zeta$ and can be increased by decreasing the value of $\tau$. Now, for any spectral efficiency below $C_{\tau,\zeta}$ our analysis reveals that the age-limited capacity can be approached as $N\xrightarrow[]{}\infty$ in which case the minimum achievable AoI is given by (\ref{min AoI_tau}). Thus, while the aggregate spectral efficiency can be made large by decreasing $\tau$, the price is an undesired  increase in the AoI. This should be expected on intuitive grounds since as $\tau$ becomes small, the users that are lucky to transmit in a given slot can be easily separated in the spatial domain by making use of a large-scale antenna array at the BS.
\end{remark}
\section{Analysis with Imperfect CSI}\label{section 5}
We incorporate imperfect CSI into our analysis by using a channel estimator and writing the error in channel estimation as $\boldsymbol{\mathcal{E}} \triangleq \mathbf{\widehat{H}}-\mathbf{H}$, where $\mathbf{\widehat{H}}$ is the channel estimate. We assume that $\mathbf{\widehat{H}}$ and $\boldsymbol{\mathcal{E}}$ are independent and $\boldsymbol{\mathcal{E}}_{i}\sim\mathcal{CN}(\mathbf{{0}}_{M},\sigma_{p}^2 \mathbf{{I}}_{M})$,  $\mathbf{\widehat{h}}_{i}\sim\mathcal{CN}(\mathbf{{0}}_{M},(1 - \sigma_{p}^2) \mathbf{{I}}_{M})$ where $\boldsymbol{\mathcal{E}}_{i}$ and $\mathbf{\widehat{h}}_{i}$ are the $i^{th}$ columns of  $\boldsymbol{\mathcal{E}}$ and $\mathbf{\widehat{H}}$ respectively and $\sigma_{p}^2$ is the mean-squared error (MSE) in channel estimation. We further assume that the columns of $\boldsymbol{\mathcal{E}}$ and $\mathbf{\widehat{H}}$ are independent. In the analysis we  append the estimation error into the noise and interference terms. The MSE, $\sigma_p^2$, of the MMSE estimator could be obtained from the state evolution of the AMP algorithm \cite{liu2018massiveII}.
\subsection{Maximal-Ratio Combining}
In the case of MRC our system model becomes:
\begin{equation}
    \mathbf{y}_n ~=~ \sum_{i=1}^{N}\mathbf{\widehat{h}}_i\varepsilon_{i}x_{i,n}~-~\sum_{i=1}^{N}\boldsymbol{\mathcal{E}}_i\varepsilon_{i}x_{i,n}~+~\mathbf{w}_n.
\end{equation}
% where $\mathbf{\widehat{h}}_{i}\sim\mathcal{CN}(\mathbf{{0}}_{M},(1 - \sigma_{p}^2) \mathbf{{I}}_{M})$ and $\boldsymbol{\mathcal{E}}_{i}$ is the $i^{th}$ column of $\boldsymbol{\mathcal{E}}$ with $\boldsymbol{\mathcal{E}}_{i}\sim\mathcal{CN}(\mathbf{{0}}_{M},\sigma_{p}^2 \mathbf{{I}}_{M})$. 
From this view, we can write the outage probability as:
% \begin{eqnarray}\label{eq:MRC_outage_csi}
%     \!\!1 - p_{e,i} = &\nonumber\\ & \!\!\!\!\!\!\!\!\!\!\!\!\!\!\!\!\!\!\!\!\!\!\!\!\!\!\! \Pr\left\{ \rho_i ~<~ \log_2\left(1+\frac{\|\mathbf{\widehat{h}}_i\|_2^4 P_i}{\|\mathbf{\widehat{h}}_i\|_2^2 \sigma_{\mathbf{w}}^2 ~+~\sum_{\substack{j=1 \\ j \neq i}}^{N}|{\mathbf{\widehat{h}}_i}^{\textsf{H}} \mathbf{\widehat{h}}_j|^2 \varepsilon_j P_j ~+~\sum_{j=1}^{N}|{\mathbf{\widehat{h}}_i}^{\textsf{H}}\boldsymbol{\mathcal{E}}_j|^{2}\varepsilon_{j}P_{j}}\right)\right\}.
% \end{eqnarray}
% % \vspace{10mm}
% \begin{equation}\label{eq:MRC_outage_csi}
%     p_{e,i} ~=~ 1 - \Pr\left\{ \rho_i ~<~ \log_2\left(1+\frac{\|\mathbf{\widehat{h}}_i\|_2^4 P_i}{\|\mathbf{\widehat{h}}_i\|_2^2 \sigma_{\mathbf{w}}^2 ~+~\sum_{\substack{j=1 \\ j \neq i}}^{N}|{\mathbf{\widehat{h}}_i}^{\textsf{H}} \mathbf{\widehat{h}}_j|^2 \varepsilon_j P_j ~+~\sum_{j=1}^{N}|{\mathbf{\widehat{h}}_i}^{\textsf{H}}\boldsymbol{\mathcal{E}}_j|^{2}\varepsilon_{j}P_{j}}\right)\right\}.
% \end{equation}
\begin{equation}\label{eq:MRC_outage_csi}
    p_{e,i} ~=~ 1 - \Pr\{ \rho_i ~<~ C_{\textrm{MRC}} \},
\end{equation}
where 
\begin{eqnarray}
    C_{\textrm{MRC}} \triangleq &\nonumber\\ & \nonumber\!\!\!\!\!\!\!\!\!\!\!\!\!\!\!\!\!\!\!\!\!\!\!\! \log_2\left(1+\frac{\|\mathbf{\widehat{h}}_i\|_2^4 P_i}{\|\mathbf{\widehat{h}}_i\|_2^2 \sigma_{\mathbf{w}}^2 ~+~\sum_{\substack{j=1 \\ j \neq i}}^{N}|{\mathbf{\widehat{h}}_i}^{\textsf{H}} \mathbf{\widehat{h}}_j|^2 \varepsilon_j P_j ~+~\sum_{j=1}^{N}|{\mathbf{\widehat{h}}_i}^{\textsf{H}}\boldsymbol{\mathcal{E}}_j|^{2}\varepsilon_{j}P_{j}}\right).
\end{eqnarray}
% $$C_{MRC} \triangleq \log_2\left(1+\frac{\|\mathbf{\widehat{h}}_i\|_2^4 P_i}{\|\mathbf{\widehat{h}}_i\|_2^2 \sigma_{\mathbf{w}}^2 ~+~\sum_{\substack{j=1 \\ j \neq i}}^{N}|{\mathbf{\widehat{h}}_i}^{\textsf{H}} \mathbf{\widehat{h}}_j|^2 \varepsilon_j P_j ~+~\sum_{j=1}^{N}|{\mathbf{\widehat{h}}_i}^{\textsf{H}}\boldsymbol{\mathcal{E}}_j|^{2}\varepsilon_{j}P_{j}} \right).$$
Conditioning on $k$ out of the total $N$ users being active and considering a symmetric system as before, i.e., $P_i$=$P ~\forall i$, and simplifying we write the conditional outage probability, $p_{e|k}$, as:
\begin{eqnarray}\label{eq:outage_conditional_csi}
    \!\!\!1-p_{e|k} = & \nonumber\\ & \!\!\!\!\!\!\!\!\!\!\!\!\!\!\!\!\!\!\!\!\!\!\!\!\!\!\!\!\!\! \Pr\left\{ \alpha_{\rho} \|\mathbf{\widehat{h}}_i\|_2^2 -\sum_{j=1}^{k}|\widetilde{\mathbf{h}}_i^{\textsf{H}} \mathbf{\widehat{h}}_j|^2 - \sum_{j=1}^{k+1}|\widetilde{\mathbf{h}}_i^{\textsf{H}} \boldsymbol{\mathcal{E}}_{j}|^2 \geq \beta \right\},
\end{eqnarray}
% \begin{equation}\label{eq:outage_conditional_csi}
%     1-p_{e|k} ~=~ \Pr\left\{ \alpha_{\rho} \|\mathbf{\widehat{h}}_i\|_2^2 ~-~ \sum_{j=1}^{k}|\widetilde{\mathbf{h}}_i^{\textsf{H}} \mathbf{\widehat{h}}_j|^2 ~-~ \sum_{j=1}^{k+1}|\widetilde{\mathbf{h}}_i^{\textsf{H}} \boldsymbol{\mathcal{E}}_{j}|^2 ~\geq~ \beta \right\},
% \end{equation}
where $\widetilde{\mathbf{h}}_i = \frac{\mathbf{\widehat{h}}_i}{\|\mathbf{\widehat{h}}_i\|_2}$. To make this convenient for use of the CLT we write (\ref{eq:outage_conditional_csi}) as:
\begin{equation}
    1-p_{e|k} ~=~ \Pr\left\{ \sum_{m=1}^{2M} S_{m} ~+~ \sum_{l=1}^{2k}V_{l} ~+~ \sum_{n=1}^{2(k+1)}T_{n} ~\geq~ \beta \right\},
\end{equation}
where $S_m\sim \Gamma(\frac{1}{2},\alpha_{\rho}(1-\sigma_{p}^2))$, $V_l\sim \Gamma(\frac{1}{2},(1-\sigma_{p}^2))$, and $T_n\sim \Gamma(\frac{1}{2},\sigma_{p}^2)$. We now apply the CLT and use the same techniques utilized previously to get the asymptotic results. We will disregard the error terms in this analysis to avoid redundancy. Applying the CLT the conditional outage probability can be written as:
\begin{equation}
    p_{e|k} \approx 1 - Q(w(k)),
\end{equation}
where $$w(k) = \frac{\beta - \alpha_{\rho}(1-\sigma_{p}^2)M+(1-\sigma_{p}^2)k+(k+1)\sigma_{p}^2}{\sqrt{\alpha_{\rho}^{2}(1-\sigma_{p}^2)^2 M+(1-\sigma_{p}^2)k+(k+1)\sigma_{p}^4}}.$$
Similar to the perfect CSI case, we make a normal approximation to the binomial distribution, $k = s\sqrt{(N-1)\tau(1-\tau)}+(N-1)\tau$, we approximate $N-1$ by $N$, and substitute $\zeta N = M$:
\begin{eqnarray}
    w(s) = &\nonumber\\& \nonumber\!\!\!\!\!\!\!\!\!\!\!\!\frac{(\beta+\sigma_{p}^2) - \alpha_{\rho}(1-\sigma_{p}^2)\zeta N+s\sqrt{N\tau(1-\tau)}+N\tau}{\sqrt{\alpha_{\rho}^{2}(1-\sigma_{p}^2)^2 \zeta N + ( (1-\sigma_{p}^2)^2 +\sigma_{p}^4) (s\sqrt{N\tau(1-\tau)}+N\tau) + \sigma_{p}^4}}.
\end{eqnarray}
% \begin{equation}
%     w(s) = \frac{(\beta+\sigma_{p}^2) - \alpha_{\rho}(1-\sigma_{p}^2)\zeta N+s\sqrt{N\tau(1-\tau)}+N\tau}{\sqrt{\alpha_{\rho}^{2}(1-\sigma_{p}^2)^2 \zeta N + ( (1-\sigma_{p}^2)^2 +\sigma_{p}^4) (s\sqrt{N\tau(1-\tau)}+N\tau) + \sigma_{p}^4}}.
% \end{equation}
Next we divide both the numerator and denominator by $\sqrt{N}$ and neglect terms that do not grow with $N$:
\begin{equation}
    w(s) = \frac{\sqrt{N}(\tau -\alpha_{\rho}(1-\sigma_{p}^2) \zeta)}{\sqrt{\alpha_{\rho}^2(1-\sigma_{p}^2)^2 \zeta +\tau ((1+\sigma_{p}^2)^2 + \sigma_{p}^4)}}.
\end{equation}
As $w(s)$ is now independent of $s$, averaging over the binomial distribution will have no affect and thus we can write the total outage probability of the system as:
\begin{equation}
    p_{e}~\approx~ Q\left( \frac{\sqrt{N}(\alpha_{\rho}(1-\sigma_{p}^2) \zeta - \tau)}{\sqrt{\alpha_{\rho}^2(1-\sigma_{p}^2)^2 \zeta +\tau ((1+\sigma_{p}^2)^2 + \sigma_{p}^4)}} \right).
\end{equation}
Similar to the perfect CSI case, we find the achievable rate for a given error probability of $\epsilon$ given by:
\begin{eqnarray}
    \rho^{\textrm{MRC}} \approx & \nonumber\\&\nonumber\!\!\!\!\!\!\!\!\!\!\!\!\!\!\!\!\!\!\!\!\!\!\!\!\log_{2}\left(1~+~\frac{1}{\frac{\tau}{\zeta}\frac{1}{(1-\sigma_{p}^2)} + \frac{Q^{-1}(\epsilon)\sqrt{\alpha_{\rho}^2(1-\sigma_{p}^2)^2 \zeta +\tau ((1+\sigma_{p}^2)^2 + \sigma_{p}^4)}}{\zeta (1+\sigma_{p}^2) \sqrt{N}}} \right),
\end{eqnarray}
% \begin{equation}
%     \rho^{\texrm{MRC}} ~\approx~ \log_{2}\left(1~+~\frac{1}{\frac{\tau}{\zeta}\frac{1}{(1-\sigma_{p}^2)} + \frac{Q^{-1}(\epsilon)\sqrt{\alpha_{\rho}^2(1-\sigma_{p}^2)^2 \zeta +\tau ((1+\sigma_{p}^2)^2 + \sigma_{p}^4)}}{\zeta (1+\sigma_{p}^2) \sqrt{N}}} \right),
% \end{equation}
which in the limit leads to:
\begin{equation}
    \rho^{\textrm{MRC}} ~=~ \log_{2}\left(1~+~\frac{\zeta}{\tau}({1-\sigma_{p}^2})\right).
\end{equation}

\subsection{Zero-Forcing}
We first start by rewriting the system model given in (\ref{received signal}) in a matrix-vector product form given by:
\begin{equation}\label{eq: matrix_vector full channel}
    \mathbf{y}_{n} = \mathbf{H}\mathbf{x}_{n}~+~\mathbf{w}_n,
\end{equation}
where $\mathbf{H} = [\varepsilon_{1}\mathbf{h}_1,\ldots,\varepsilon_{N}\mathbf{h}_N]$ and $\mathbf{x}_{n} = [x_{1,n},\ldots,x_{N,n}]^{\textsf{T}}$. Conditioned on $k$ out of the $N$ total users being active we can write ZF receiver as:
\begin{equation}
    \mathbf{W}_{\mathrm{ZF}}^{\textsf{H}} = \frac{1}{\sqrt{P}}(\mathbf{H}^{\textsf{H}}\mathbf{H})^{-1}\mathbf{H}^{\textsf{H}},
\end{equation}
where now $\mathbf{H}$ is only composed of the channels of the $k$ active transmitters. Applying this linear receiver to $\mathbf{y}_{n}$ we have:
\begin{equation}
    \mathbf{W}_{\mathrm{ZF}}^{\textsf{H}} \mathbf{y}_{n} = \mathbf{s}_{n}~+~\mathbf{v}_n,
\end{equation}
where now $\mathbf{s}_{n}\sim\mathcal{{CN}}(\mathbf{{0}}_{k},\mathbf{I}_{k})$ and $\mathbf{v}_{n}$ is coloured noise with
$\mathbb{E}[\mathbf{v}_{n}\mathbf{v}_{n}^{\textsf{H}}|\mathbf{H}] = \beta(\mathbf{H}^{\textsf{H}}\mathbf{H})^{-1}$. In this case, the $i^{th}$ user will see a scalar channel and the maximum achievable rate over this channel will be given by $\log_{2}(1+\textrm{SNR}_{i}^{\textrm{ZF}})$
where $\textrm{SNR}_{i}^{\textrm{ZF}}$ is given by:
\begin{equation}
    \textrm{SNR}_{i}^{\mathrm{ZF}} = \frac{\mathbb{E}[\mathbf{s}_{n}\mathbf{s}_{n}^{\textsf{H}}]_{i,i}}{\mathbb{E}[\mathbf{v}_{n}\mathbf{v}_{n}^{\textsf{H}}|\mathbf{H}]_{i,i}} = \frac{1}{\beta(\mathbf{H}^{\textsf{H}}\mathbf{H})_{i,i}^{-1}}.
\end{equation}
$\textrm{SNR}_{i}^{\textrm{ZF}}$ is found to follow a Chi-Squared distribution with $2(M-k+1)$ degrees of freedom (DOF). For the case that no other users are active besides user $i$, $\textrm{SNR}_{i}^{\textrm{ZF}}$ follows a Chi-Squared distribution with $2M$ DOF \cite{heath2018foundations}. Therefore, averaging over all of the users we can write the error probability as:
\begin{eqnarray}\label{eq:ZF_finite_pe}
p_{e} ~=~ (1-\tau)^{N-1}\int_{0}^{2^{\rho}-1}f_{\chi_{2M}^2}(\eta)d\eta &\\ & \!\!\!\!\!\!\!\!\!\!\!\!\!\!\!\!\!\!\!\!\!\!\!\!\!\!\!\!\!\!\!\!\!\!\!\!\!\!\!\!\!\!\!\!\!\!\!\!\!\!\!\!\!\!\!\!\!\!\!\!\!\!\!\!\!\!\!\!\!\!\!\!\!\!\!\!\!\!\!\!\!\!\!\!\!\!\!\!\!\!\! +\mathlarger{\sum_{k=1}^{N-1}}{N-1\choose k}\tau^{k}(1-\tau)^{N-1-k}\mathlarger{\int_{0}^{2^{\rho}-1}}f_{\chi_{2(M-k+1)}^2}(\eta)d\eta \nonumber,
\end{eqnarray}
% \begin{equation}\label{eq:ZF_finite_pe}
%     p_{e}~ =~ (1-\tau)^{N-1}\int_{0}^{2^{\rho}-1}f_{\chi_{2M}^2}(\eta)d\eta ~+~ \sum_{k=1}^{N-1}{N-1\choose k}\tau^{k}(1-\tau)^{N-1-k}\int_{0}^{2^{\rho}-1}f_{\chi_{2(M-k+1)}^2}(\eta)d\eta,
% \end{equation}
from which the AoI simply follows from (\ref{definition of network AoI}).\\

\noindent To analyze our system asymptotically under ZF we write our system model as:
\begin{equation}\label{eq:ZF_CSI_model}
    \mathbf{y}_{n} = \sqrt{P(1-\sigma_{p}^2)}\mathbf{\widehat{H}}\mathbf{s}_{n}~-~\sqrt{P}\boldsymbol{\mathcal{E}}\mathbf{s}_{n}~+~\mathbf{w}_n.
\end{equation}
Conditioning on $k$ out of the $N$ users being active the ZF filter can be written as:
\begin{equation}
    \mathbf{W}_{\mathrm{ZF}}^{\textsf{H}} = \frac{1}{\sqrt{P(1-\sigma_{p}^2)}}(\mathbf{\widehat{H}}^{\textsf{H}}\mathbf{\widehat{H}})^{-1}\mathbf{\widehat{H}}^{\textsf{H}},
\end{equation}
where now $\mathbf{\widehat{H}}$ is only composed of the channels of the $k$ active transmitters. Applying this linear receiver to $\mathbf{y}_{n}$ we have:
\begin{equation}
    \mathbf{W}_{\mathrm{ZF}}^{\textsf{H}} \mathbf{y}_{n} = \mathbf{s}_{n}~+~\mathbf{v}_{n},
\end{equation}
where now $\mathbf{s}_{n}\sim\mathcal{{CN}}(\mathbf{{0}}_{k},\mathbf{I}_{k})$ and $\mathbf{v}_{n}$ is coloured noise with the estimation error term, i.e:
\begin{equation}
    \mathbf{v}_{n} ~=~ -\sqrt{P}\mathbf{W}_{\mathrm{ZF}}^{\textsf{H}}\boldsymbol{\mathcal{E}}\mathbf{s}_{n}~+~\mathbf{W}_{\mathrm{ZF}}^{\textsf{H}}\mathbf{w}_n.
\end{equation}
In this case, the $i^{th}$ user will see a scalar channel and the maximum achievable rate over this channel will be given by $\log_{2}(1+\textrm{SNR}_{i}^{\textrm{ZF}})$
where $\textrm{SNR}_{i}^{\textrm{ZF}}$ is given by:
\begin{equation}
    \textrm{SNR}_{i}^{\mathrm{ZF}} = \frac{\mathbb{E}[\mathbf{s}_{n}\mathbf{s}_{n}^{\textsf{H}}]_{i,i}}{\mathbb{E}[\mathbf{v}_{n}\mathbf{v}_{n}^{\textsf{H}}]_{i,i}} = \frac{1}{\mathbb{E}[\mathbf{v}_{n}\mathbf{v}_{n}^{\textsf{H}}]_{i,i}}.
\end{equation}
It is easy to show that:
\begin{equation}
    \mathbb{E}[\mathbf{v}_{n}\mathbf{v}_{n}^{\textsf{H}}]~=~\left( \frac{k \sigma_{p}^2 +  \beta}{1-\sigma_{p}^2}\right)(\mathbf{\widehat{H}}^{\textsf{H}}\mathbf{\widehat{H}})^{-1}.
\end{equation}
Therefore, we can now write the conditional outage probability as:
\begin{equation}
    p_{e|k}~=~\Pr\left\{  \frac{1}{\left( \frac{k \sigma_{p}^2 +  \beta}{1-\sigma_{p}^2}\right)(\mathbf{\widehat{H}}^{\textsf{H}}\mathbf{\widehat{H}})_{i,i}^{-1} } \leq \frac{1}{\alpha_{\rho}} \right\}
\end{equation}
For the asymptotic analysis we let $M$ and $N$ grow large while holding their ratio, $\zeta = M/N$, constant. In the case of ZF we additionally have the condition that $\zeta \geq 1$. We can see the first term in (\ref{eq:ZF_finite_pe}) goes to $0$ as $N$ grows large and we focus on the second term. We note that the RV:
\begin{equation}
    Z = \frac{\chi_{2(M-k+1)}^2 - (M-k+1)/\left( \frac{k \sigma_{p}^2 +  \beta}{1-\sigma_{p}^2}\right)}{\sqrt{M-k+1}/\left( \frac{k \sigma_{p}^2 +  \beta}{1-\sigma_{p}^2}\right)},
\end{equation}
tend to a normal distributions by the central limit theorem as the number of DOF grow large. Therefore we can rewrite the second integral in (\ref{eq:ZF_finite_pe}) as:
\begin{equation}\label{U integral}
    \int_{0}^{2^{\rho}-1}f_{\chi_{2(M-k+1)}^2}(\eta)d\eta ~\xrightarrow[]{}~ \int_{-\sqrt{M-k+1}}^{\lambda_{M,k}}f_{Z}(z)dz,
\end{equation}
where:
\begin{equation}
    \lambda_{M,k}~ =~ \frac{ k \sigma_{p}^2 +  \beta}{\alpha_{\rho}(1-\sigma_{p}^2)\sqrt{M-k+1}}-\sqrt{M-k+1}.\nonumber
\end{equation}
% \begin{equation}
%     \int_{0}^{2^{\rho}-1}f_{\chi_{2(M-k+1)}^2}(\eta)d\eta ~\xrightarrow[]{}~ \int_{-\sqrt{M-k+1}}^{\frac{\left( \frac{k \sigma_{p}^2 +  \beta}{1-\sigma_{p}^2}\right)/\alpha_{\rho}}{\sqrt{M-k+1}}-\sqrt{M-k+1}}f_{Z}(z)dz.
% \end{equation}
We can write the integral in (\ref{U integral}) as a difference of two Q-functions given by the upper and lower bounds of the integral, i.e:
\begin{equation}
    \int_{-\sqrt{M-k+1}}^{\lambda_{M,k}}f_{Z}(z)dz =Q\left(-\sqrt{M-k+1} \right) - Q(\lambda_{M,k}).
\end{equation}
%\begin{eqnarray}
%     \int_{-\sqrt{M-k+1}}^{U(M,k)}f_{Z}(z)dz ~~= & \nonumber\\ & \!\!\!\!\!\!\!\!\!\!\!\!\!\!\!\!\!\!\!\!\!\!\!\!\!\!\!\!\!\!\!\!\!\!\!\!\!\!\!\!\!\!\!\!\!\!\!\!\!\!\!\!\!\!\!\!\!\!\!\!\!\!\!\!\!\!\!\!\!\!\!\!\!\!\!\!\!\!\!\!\!\!\!\!\!\!\!\! Q\left(-\sqrt{M-k+1} \right) - Q\left(\frac{ \left( \frac{k \sigma_{p}^2 +  \beta}{1-\sigma_{p}^2}\right) / {\alpha_{\rho}} -(M-k+1)}{\sqrt{M-k+1}} \right).
% \end{eqnarray}
% \begin{equation}
%     \int_{-\sqrt{M-k+1}}^{\frac{\left( \frac{k \sigma_{p}^2 +  \beta}{1-\sigma_{p}^2}\right)/\alpha_{\rho}}{\sqrt{M-k+1}}-\sqrt{M-k+1}}f_{Z}(z)dz = Q\left(-\sqrt{M-k+1} \right)~-~ Q\left(\frac{ \left( \frac{k \sigma_{p}^2 +  \beta}{1-\sigma_{p}^2}\right) / {\alpha_{\rho}} -(M-k+1)}{\sqrt{M-k+1}} \right).
% \end{equation}
As we let $M$ grow large for a fixed $k$ the term on the left goes to 1. Making the substitution $M = \zeta N$, approximating $N-1$ by $N$, and approximating the binomial distribution using the CLT through $k= s\sqrt{(N-1)\tau(1-\tau)}+(N-1)\tau$, we can write the total probability error as:
\begin{equation}
    p_{e,i} ~=~ 1 ~-~ \frac{1}{\sqrt{2 \pi}}\int_{-\infty}^{\infty}Q(w(s))e^{-\frac{s^2}{2}}ds,
\end{equation}
where $w(s)$ is given by:
\begin{eqnarray}
    \!\!\!\!\!\!\!\!\!\!\!\!w(s) = &\nonumber\\& \!\!\!\!\!\!\!\!\!\!\!\!\!\!\!\!\!\! \frac{\left(\frac{((2^{\rho}-1) \beta)}{1-\sigma_{p}^2} -1\right)- (s\sqrt{N\tau(1-\tau)}+N\tau)\left(\frac{(2^{\rho}-1) \sigma_{p}^2}{(1-\sigma_{p}^2)} +1\right) - \zeta N}{\sqrt{\zeta N-(s\sqrt{N\tau(1-\tau)}+N\tau)+1}}.
\end{eqnarray}
% \begin{equation}
%     w(s) = \frac{\left(\frac{((2^{\rho}-1) \beta)}{1-\sigma_{p}^2} -1\right)- (s\sqrt{N\tau(1-\tau)}+N\tau)\left(\frac{(2^{\rho}-1) \sigma_{p}^2}{(1-\sigma_{p}^2)} +1\right) - \zeta N}{\sqrt{\zeta N-(s\sqrt{N\tau(1-\tau)}+N\tau)+1}}.
% \end{equation}
Dividing the numerator and denominator of $w(s)$ by $\sqrt{N}$ and neglecting terms that do not grow with $N$ we get:
\begin{equation}
    w(s) = \frac{\sqrt{N}\left( \tau\left(\frac{(2^{\rho}-1 )\sigma_{p}^2}{(1-\sigma_{p}^2)} +1\right) - \zeta  \right)}{\sqrt{\zeta - \tau}}.
\end{equation}
Finally, since this term does not not depend on $s$ we take it out of the integral and write the asymptotic packet error probability as:
\begin{equation}
    p_{e,i} = Q\left( \frac{\sqrt{N}\left(\zeta - \tau\left(\frac{(2^{\rho}-1 )\sigma_{p}^2}{(1-\sigma_{p}^2)} +1\right) \right)}{\sqrt{\zeta - \tau}}\right).
\end{equation}
Similarly we find the achievable rate for a given error probability of $\epsilon$:
\begin{equation}
    \rho^{\textrm{ZF}} \approx \log_{2}\left(1~+\frac{1-\sigma_{p}^2}{\sigma_{p}^2}\left( \frac{\zeta}{\tau} - 1 - \frac{Q^{-1}(\epsilon) \sqrt{\zeta-\tau}}{\tau \sqrt{N}} \right)~ \right),
\end{equation}
which in the limit leads to:
\begin{equation}
    \rho^{\textrm{ZF}} ~=~ \log_{2}\left(1~+\frac{1-\sigma_{p}^2}{\sigma_{p}^2}\left(\frac{\zeta}{\tau} - 1\right)~ \right).
\end{equation}

%%%%%%%%%%%%%%%%%%%%%%%%%%%%%%%%%%%%%%%%%%%%%%%%%%%%%%%%%%%%%%%%%%%%%%%%%%%%%%%%%%%%%%%%%%%%%%%%%}

\section{Numerical Results and Discussion}\label{section_6}
We now illustrate the results found in the previous sections and further compare recent URA schemes against our bound. 
\subsection{Trade-Off Between AoI and Spectral Efficiency}
In order to determine the direct relationship between the spectral efficiency and the AoI, we start by re-writing (\ref{Clean AoI}) as a function of the spectral efficiency for a given probability of error, $p_{e} = \epsilon$. In this case, the AoI reduces simply to:
\begin{equation}\label{AoI as for fixed pe}
    \Delta(\zeta,N,\rho,\tau) ~=~ \frac{1}{\tau_{\epsilon}(1-\epsilon)} ~+~ \mathcal{O}\left(\frac{1}{\sqrt{N}}\right),
\end{equation}
where $\tau_{\epsilon}$ is found by solving for $\tau$ in the Q-function in (\ref{Clean AoI}) and is given by
\begin{eqnarray}\label{constant Pe tau}%we choose the (-) so pe = epsilon not 1-epsilon
    \tau_{\epsilon}=\alpha_{\rho}\zeta+ \frac{Q^{-1}(\epsilon)^{2}}{2N}&\nonumber\\&   \!\!\!\!\!\!\!\!\!\!\!\!\!\!\!\!\!\!\!\!\!\!\!\!\!\!\!\!\!\!\!\!\!\!\!\! - \sqrt{\left(\alpha_{\rho}\zeta+\frac{Q^{-1}(\epsilon)^{2}}{2N}\right)^{2}+ \alpha_{\rho}^{2}\left(\frac{Q^{-1}(\epsilon)^{2}\zeta}{N}-\zeta^{2}\right)}.
\end{eqnarray}
The identity in (\ref{constant Pe tau}) is valid for $\rho \geq \rho_\textrm{min}(\epsilon)$ where $\rho_\textrm{min}(\epsilon)$ is given in (\ref{Supremum}) evaluated at $\tau = 1$ ({this is seen when solving for $\tau$ from the Q-function and observing where the solution is valid}).
In Fig. \ref{fig:AoI_vs_rho}, we plot (\ref{AoI as for fixed pe}), ignoring the error terms, for $p_{e}=10^{-5}$, $\zeta = 0.3$, and $N = \{100,500,1000\}$. We also plot the case of infinite number of users with $p_{e} = 0$ 
\begin{figure}[h!]
    \centering
    \includegraphics[width=1\linewidth]{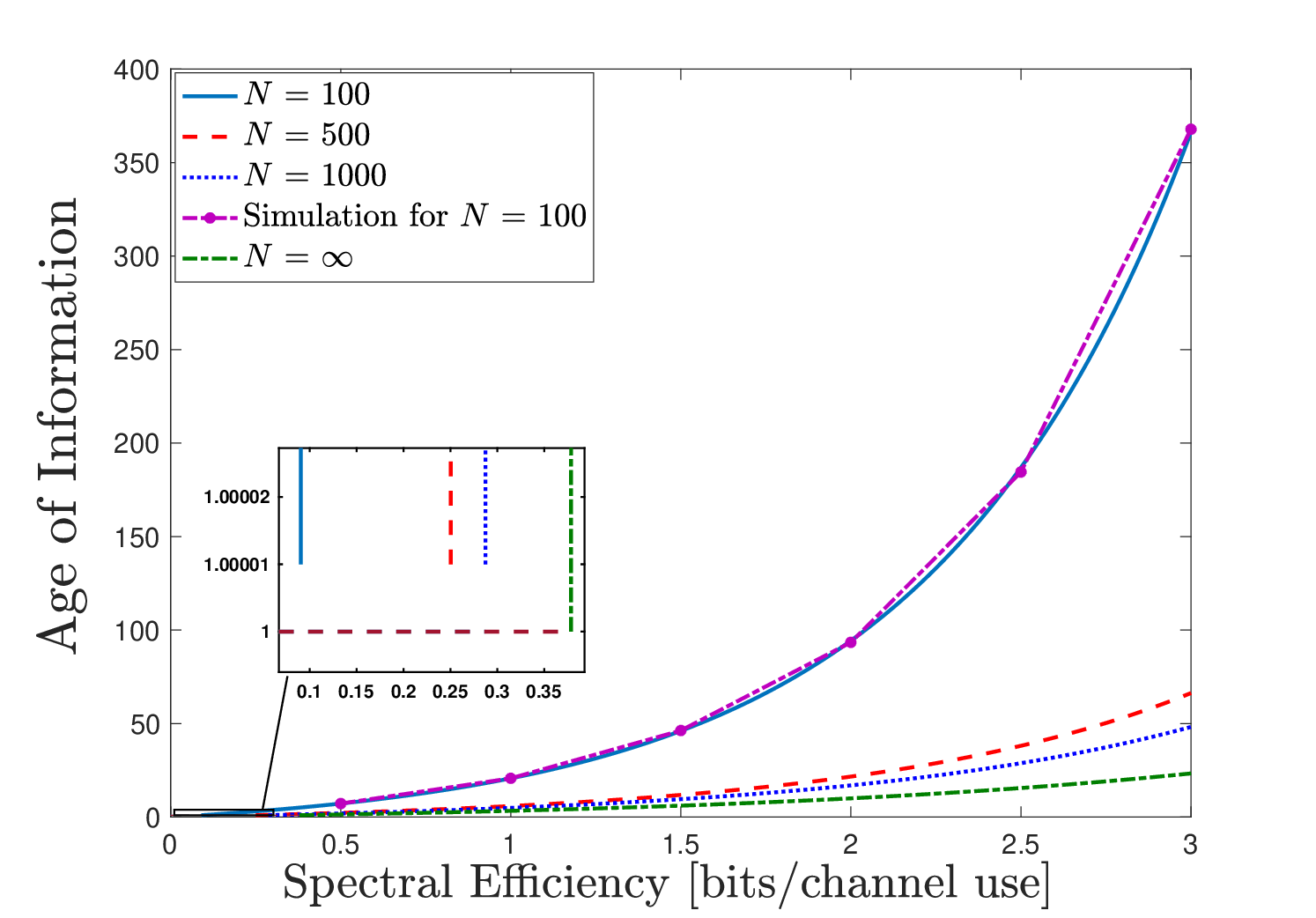} 
    %\hspace{-20mm}
    % \vskip 0.2cm
    \caption{AoI vs Spectral Efficiency at $p_{e} = 10^{-5}$ and $\zeta$ = 0.3. In the simulation, for $N=100$ we set the number of slots to be $10^{5}$.}  
    \label{fig:AoI_vs_rho}
\end{figure}
% \begin{figure}[h!]
%     \centering
%     \includegraphics[width=1\linewidth]{AoI_vs_rho_with_sim.eps} 
%     \vskip 0.2cm
%     \caption{AoI vs Spectral Efficiency at $p_{e} = 10^{-5}$ and $\zeta$ = 0.3. In the simulation, for $N=100$ we set the number of slots to be $10^{5}$.}  
%     \label{fig:AoI_vs_rho}
% \end{figure}
whose curve is obtained by taking $N\to\infty$ in (\ref{constant Pe tau}), i.e.:
\begin{equation}\label{optimum operating point}
    \lim_{N \to \infty}\tau_{\epsilon} ~=~ \alpha_{\rho}\zeta.
\end{equation}
 {Each spectral efficiency on the $N=\infty$ curve is the age-limited capacity for a given set of $\tau$ and $\zeta$.} In both the finite- and infinite-number-of-users scenarios, the AoI is minimized by setting $\tau$ to 1,  {as seen in (\ref{AoI as for fixed pe})}. In the finite case, however, the minimum AoI is limited by the probability of error and is given by:
\begin{equation}\label{minimum_age_epsilon}
    \Delta_{\epsilon} ~=~ \frac{1}{1-\epsilon} ~+~\mathcal{O}\left(\frac{1}{\sqrt{N}}\right).
\end{equation}
On the other hand, in the infinite-number-of-users regime the probability of error is driven to zero and the AoI takes the minimum value possible, $\Delta = 1$. This is  clearly seen in the zoomed portion of Fig. \ref{fig:AoI_vs_rho}. We also add a simulation for $N = 100$. In our simulation we choose a large number of time slots, $10^{5}$, and simulate the network AoI. Indeed, we can see from Fig. \ref{fig:AoI_vs_rho} that the error term in our approximations is small.
\subsection{Performance under Imperfect CSI} We compare the achievable rates for MRC and ZF for a given MSE in channel estimation in Fig. \ref{fig:rho_vs_mse}. Indeed as expected the ZF receiver performs better given that the channel estimate is good. However, the MRC receiver is more robust and is not as sensitive to imperfect channel estimates. Additionally, the constraint $\zeta \leq 1$ for ZF heavily limits its use in mMTC where the total number of users is much larger than the number of BS antennas. Similar results follow for both AoI and Outage performance.
\begin{figure}[h!]
    \centering
    \includegraphics[width=1\linewidth]{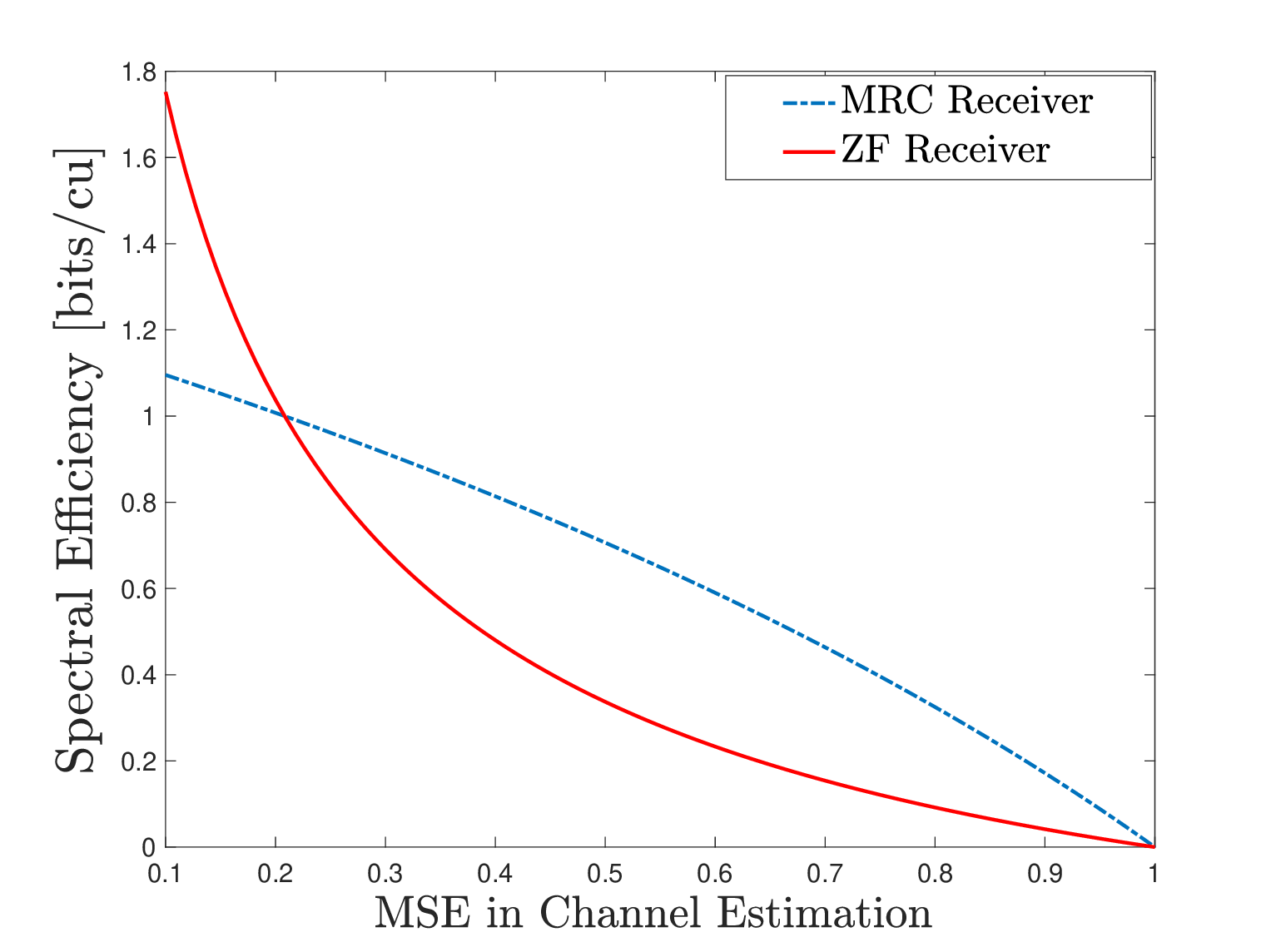} 
    %\vskip 0.1cm
    \caption{Achievable rate against Channel Estimation Error, here we use $\zeta = 1.2$, and $\tau = 0.95$.}
    \label{fig:rho_vs_mse}
\end{figure}

\subsection{Application to Unsourced  Random Access (URA)}
The URA paradigm, initially analyzed in \cite{polyanskiy2017perspective}, in which the base station is tasked with providing multiple access to a large number of uncoordinated users, has attracted considerable attention. In \cite{polyanskiy2017perspective}, the random coding achievability bound was derived and compared against popular multiple-access schemes. A number of algorithmic solutions for URA were proposed in \cite{amalladinne2018coupled,calderbank2020chirrup,pradhan2019joint,pradhan2020polar}. However, all of the above theoretical and algorithmic works focused on the case of a single receive antenna at the BS. To date, the only two algorithmic solutions for the URA paradigm that have so far investigated the use of the massive MIMO technology are \cite{shyianov2020massive,fengler2019massive}  whose performances have not been gauged against any achievability bound. We will refer to these two massive MIMO-based URA schemes in \cite{shyianov2020massive,fengler2019massive} as ``clustering-based" and ``covariance-based", respectively. 

In Fig. \ref{fig:Algorithm Comparisons}, we use computer simulations to compare both schemes to the new achievability bound established in \textbf{Theorem \ref{supremum theorem}} as well as to the exact expression established in (\ref{Non-asymptotic Pi}). 
\begin{figure}[h!]
    \centering
    \includegraphics[width=1\linewidth]{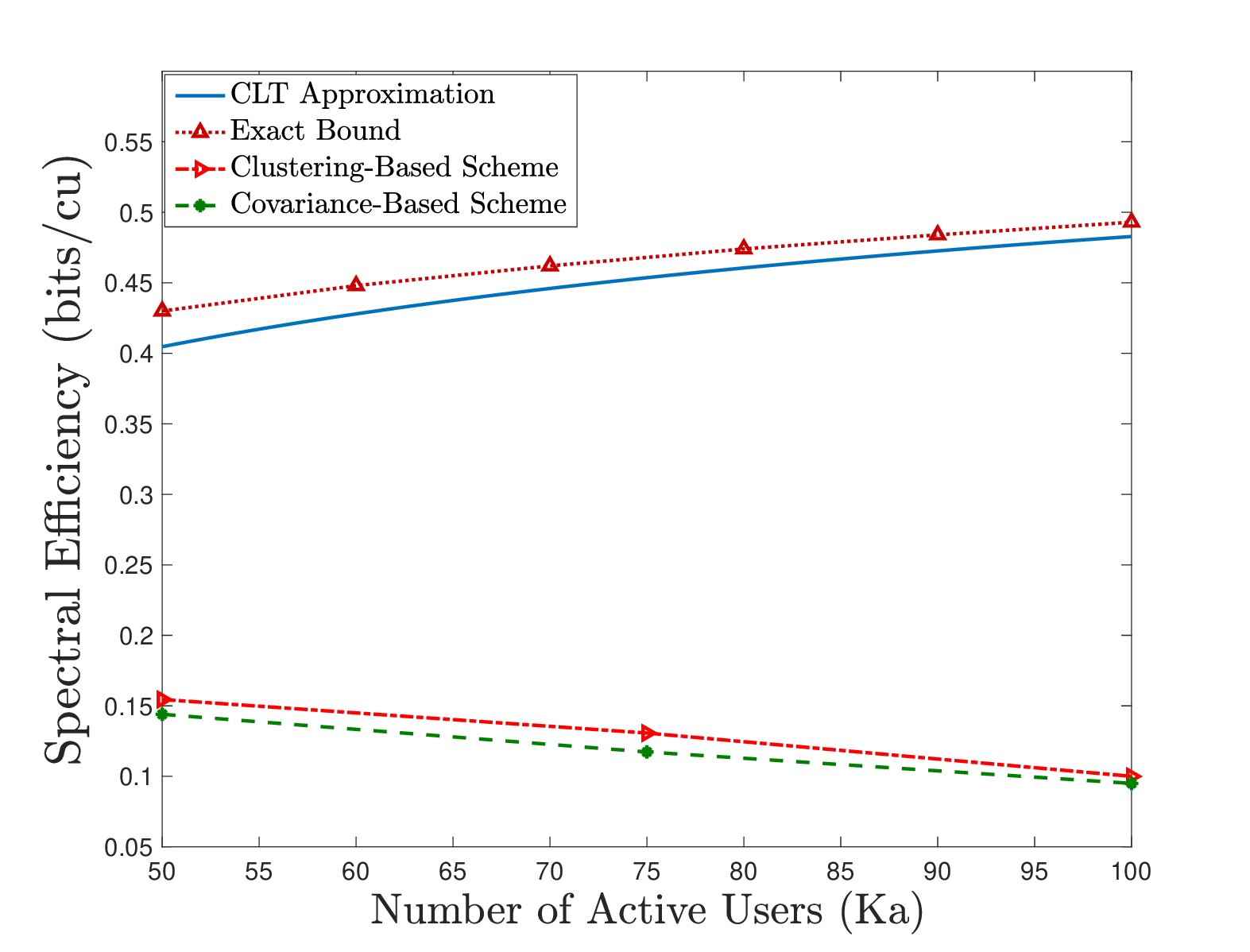} 
    %\vskip 0.1cm
    \caption{Performance of two recent URA schemes against the newly established achievability bound in (\ref{Supremum}) with $p_e = 10^{-2}$.}  
    \label{fig:Algorithm Comparisons}
\end{figure}
For both the schemes, we fix the bandwidth to $W = 10$ MHz and the noise variance to $\sigma_{w}^2 = 10^{-19.9}\times W$  [Watts] and then calculate the required transmit power that yields the SNR, $\beta^{-1} = 30$ dB. For the clustering-based URA scheme, we use a Gaussian prior for HyGAMP-based compressed sensing (CS) and communicate $B = 102$ information bits per user/packet over $L = 6$ slots. For the covariance-based scheme, we fix the number of information bits per user/packet to $B = 104$ bits which are communicated over $L = 17$ slots. The parity bit allocation for the outer tree code is set to $p = [0,8,8,\ldots,14]$. We also use $J = 14$ coded bits per slot which leads to the total rate of the outer code $R_\textrm{out} = 0.437$. For both the schemes, we simulate $3$ data points with $K_{a} = [50,75,100]$ active users and $M = [30,45,60]$ antennas at the base station in which case the achievable spectral efficiency in (\ref{Age_Capacity_with_tau}) is given by $\log_2\left(1+\frac{M}{K_a}\right)$. Note that even though the achievable spectral efficiency does not depend on the total number of users, as is the case in \cite{polyanskiy2017perspective}, the AoI does. In fact, as the total number of users increases, the AoI grows unbounded for any fixed number of active users. 

In the plots of Fig. \ref{fig:Algorithm Comparisons}, apart from the gap between the newly established bound and the existing algorithmic solutions, it is seen that the achievable spectral efficiency of both schemes decreases as the number of active users increases. Actually, it should be possible to rigorously prove this limitation for any CS-based decoding scheme. Roughly speaking, as the number of active users grows, increasing the per-user spectral efficiency requires one to decrease the blocklength thereby rendering the CS-based support recovery task more challenging. In fact, the fundamental limitation of support recovery requires (see Chapter 7 of \cite{wainwright2019high}) the blocklength $L = \mathcal{O}(K_a\log(\frac{2^{B/L}}{K_a}))$ to scale a little faster than the number of active users with a single antenna at the BS. On the contrary, in \cite{fengler2019massive,fengler2021non}, it was shown that the covariance-based URA scheme with a large-scale antenna array at the BS can recover the support perfectly as long as $K_a\log^2\big(\frac{2^{B/L}}{K_a}\big) = \mathcal{O}(L^2)$ and $\frac{K_a}{M} = o(1)$. However, in this case, the achievable spectral efficiency goes to infinity and the achievable performance with respect to \textbf{Theorem \ref{supremum theorem}} has to be investigated carefully. In particular, a sharper characterization of the $\frac{K_a}{M} = o(1)$ term is required. 
%%%%%%%%%%%%%%%%%%%%%%%%%%%%%%%%%%%%%%%%%%%%%%%%%%%%%%%%%%%%%%%%%%%%%%%%%%%%%%%%%%%%%%%%%%%%%%%%%
\section{Conclusion}\label{Section Conclusion}
We have established achievability and converse results in the $K_{a}$-user GMAC with a large-scale antenna array at the BS. We have defined the age-limited capacity as the maximum spectral efficiency achievable such that the AoI is finite in asymptotic system limits. In this case, we have shown that in order to minimize the system AoI all devices must be active in every transmission period.
This is also the case in finite system sizes in which the AoI is, however, limited by the probability of error. We have used our bound to compare the two recent massive MIMO URA algorithms, thereby revealing a huge gap between their performance and the overall spectral efficiency that can be potentially achieved in practice. In future work, considering the overloaded system \cite{liu2018massive}, $\zeta/\tau < 1$, it is desirable to do some scheduling in order to control the inter-user interference. One approach could be to subdivide each slot into $J$ scheduling intervals. Then the asymptotic joint optimization of the spectral efficiency as well as the AoI could be performed over the number of scheduling intervals.
%%%%%%%%%%%%%%%%%%%%%%%%%%%%%%%%%%%%%%%%%%%%%%%%%%%%%%%%%%%%%%%%%%%%%%%%%%%%%%%%%%%%%%%%%%%%%%%%%
%\newpage
\appendix

\subsection{Derivation of the AoI}\label{appendix:Age}
Recall from (\ref{Age of Information}) the definition of AoI. The latter is a limit of a time-average of the age sample function as the time horizon grows large. Instead of computing the integral in (\ref{Age of Information}) directly, we can express it as a function of the inter-update times, as shown pictorially in Fig. \ref{fig:Decomposed AoI}.
\begin{figure}[H]
    \centering
    \includegraphics[width=0.9\linewidth]{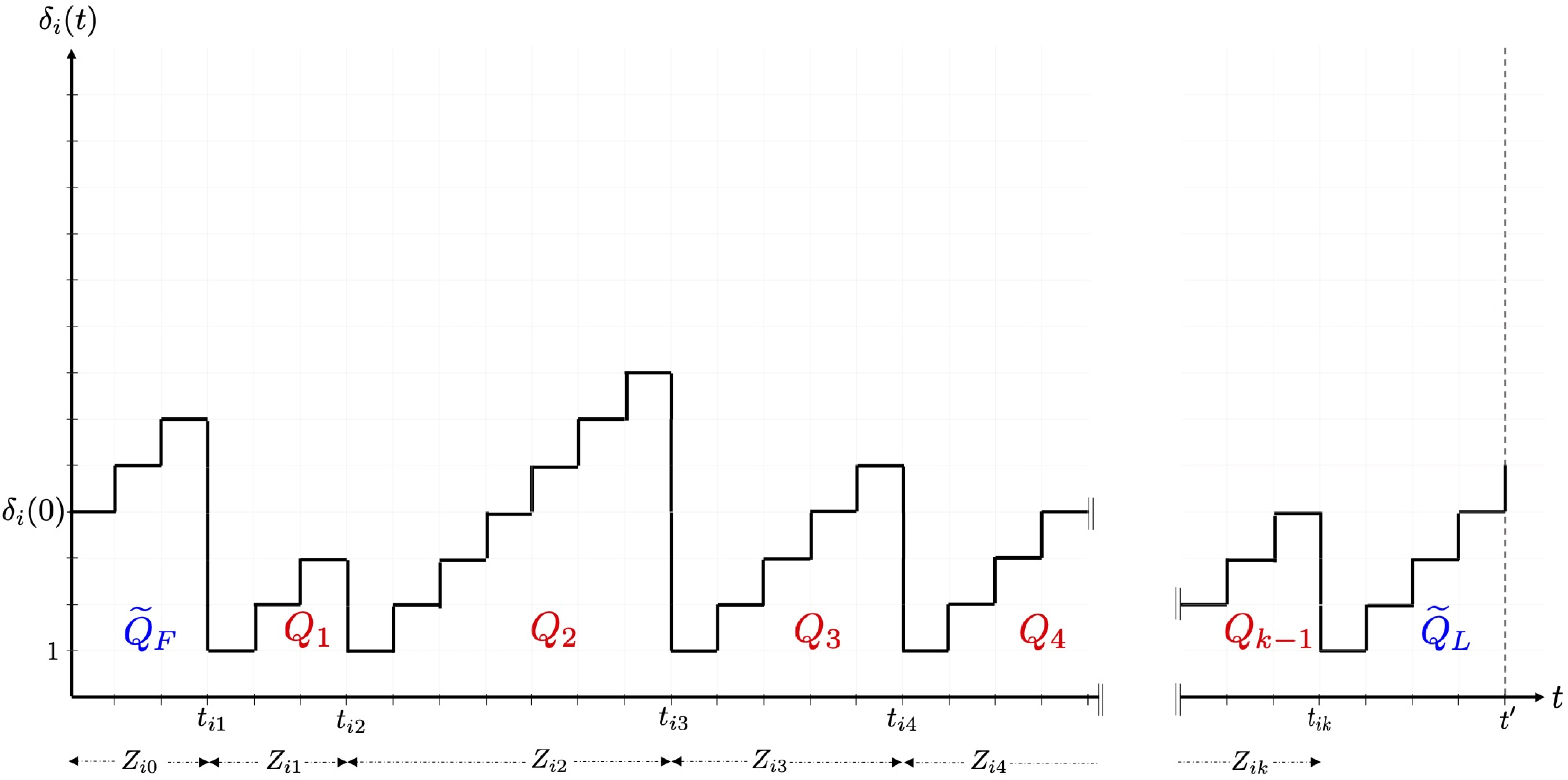}
    \vskip 0.1cm
    \caption{Decomposition of a sample function of the AoI.}  
    \label{fig:Decomposed AoI}
\end{figure}
The area below the sample function shown in Fig. \ref{fig:Decomposed AoI} consists of a rectangular base of width $t'$ and unit height, jagged triangular structures, and boundary pieces (both above the base rectangle). We denote the area of the triangular pieces by $Q_{l}$ and the first and last boundary pieces by $\widetilde{Q}_{F}$ and $\widetilde{Q}_{L}$, respectively. Now (\ref{Age of Information}) can be re-written as:
\begin{equation}\label{without_time_average}
    \int_{0}^{t'}\delta_i(t)dt ~=~ t'+\widetilde{Q}_{F}+\widetilde{Q}_{L}+\sum_{k=1}^{N(t')-1}Q_{k},
\end{equation}
where $N(t')$ denotes the number of arrivals by time $t'$. The $Q_l$'s can be written in terms of the inter-update intervals as:
\begin{equation}
    Q_{l} = \frac{Z_{il}(Z_{il}-1)}{2}.
\end{equation}
Dividing the right-hand side of (\ref{without_time_average}) by $t'$ and taking the limit (as $t'$ goes to $+\infty$) we have:
\begin{equation}\label{almost done}
    \Delta_{i} = \lim_{t' \to \infty}\left(1+\frac{\widetilde{Q}_{F}+\widetilde{Q}_{L}}{t'}+\sum_{k=1}^{N(t')-1}\frac{Z_{ik}(Z_{ik}-1)}{2t'}\right).
\end{equation}
The first term can be taken out of the limit and the second term goes to 0 WP1. We re-write the last term in (\ref{almost done}) as follows:
\begin{equation}
    \lim_{t'\to\infty}\left(\frac{N(t')}{2t'}-\frac{1}{2t'}\right)\sum_{k=1}^{N(t')-1}\frac{Z_{ik}(Z_{ik}-1)}{N(t')-1}.
\end{equation}
As the Bernoulli process is a renewal process, we have $\frac{N(t')}{2t'}\to\frac{1}{2\mathbb{E}[Z_{i}]}$ WP1 \cite{gallager2013stochastic}. Also, from the strong law of large numbers (SLLN) the sum in the limit converges to $\mathbb{E}[Z_{i}^2]-\mathbb{E}[Z_{i}]$ WP1. Therefore, we have:
\begin{equation}
    \Delta_{i} ~=~ \frac{\mathbb{E}[Z_{i}^2]}{2\mathbb{E}[Z_{i}]}+\frac{1}{2}.
\end{equation}

%%%%%%%%%%%%%%%%%%%%%%%%%%%%%%%%%%%%%%%%%%%%%%%%%%%%%%%%%%%%%%%%%%%%%%%%%%%%%%%%%%%%%%%%%%%%%%%%%

%As the random variables in the second term of (\ref{Explicit Total Probability}) are not identically distributed we must use a variant of the CLT, called Lyapunov's CLT. The Lyapunov CLT holds as long as the RVs involved are independent and they satisfy the Lyapunov condition \cite{patrick1995probability}, i.e., if for some $\delta$
%\begin{equation}\label{Lyapunovs condition}
 %   \lim_{n \to \infty}\frac{1}{s_{n}^{2+\delta}}\sum_{i=1}^{n}\mathbb{E}\left[|X_{i}-\mu_{i}|^{2+\delta}\right] = 0,
%\end{equation}
%where $s_{n}$ $\triangleq$ $\sum_{i=1}^{n}\sigma_{i}^2$. In which our case we have $s_{n}$ = $\sqrt{M\alpha_{R}^{2}+k}$ and the sum in (\ref{Lyapunovs condition}) can be split into the two different types of RVs, i.e., the $S_{i}$s and $V_{j}$s in (\ref{\CLT convergence}). We choose $\delta$ to be 2 and write
%\begin{equation}\label{explicit lyapunov}
 %   F_{M,k} = \frac{1}{s_{n}^4}\left[\sum_{i=1}^{2M}\mathbb{E}[|S_{i}-\mu_{i}|^4]+\sum_{j=1}^{2k}\mathbb{E}[|V_{j}-\mu_{j}|^4] \right].
%\end{equation}
%Computing the fourth central moments of $S_{i}$ and $V_{j}$ is a trivial task and we obtain, $\mathbb{E}[|S_{i}-\mu_{i}|^4]$ = $\frac{15}{4}(\alpha_{R})^4$ and $\mathbb{E}[|V_{j}-\mu_{j}|^4]$ = $\frac{47}{4}$. These both do not grow with $M$ and therefore, $F_{M,k}$ $\to$ 0 as $M$ $\to$ $\infty$ and we satisfy Lyapunov's conditions.

\subsection{Proofs of the Various Approximations}\label{Appendix:CLT}
\begin{lemma}\label{Lemma_1}
The inverse CDF of the sum of random variables, $\sum_{m=1}^{2M}S_{m}+\sum_{l=1}^{2k}V_{l}$, converges uniformly to the standard normal inverse CDF, where the convergence rate is bounded above by $\mathcal{O}\left(\frac{1}{\sqrt{M+k}}\right)$,  
\begin{equation}
    \Pr\left\{\sum_{m=1}^{2M}S_{m}+\sum_{l=1}^{2k}V_{l} \geq \beta\right\} ~=~ Q(w(k)) + \mathcal{O}\left(\frac{1}{\sqrt{M+k}}\right),
\end{equation}
in which $w(k) = \frac{\beta - \alpha_{\rho}M+k}{\sqrt{\alpha_{\rho}^{2}M+k}}$.
\end{lemma}
\begin{proof}
Let $B_{2M+2k}$ be the following normalized RV:
\begin{eqnarray}
    B_{2M+2k} ~~\triangleq &\nonumber\\\nonumber&
    \!\!\!\!\!\!\!\!\!\!\!\!\!\!\!\!\!\!\!\!\!\frac{1}{\sqrt{\alpha_{\rho}^{2}M~+~k}}\left[\sum_{m=1}^{2M}S_{m}+\sum_{l=1}^{2k}V_{l}~-~(M\alpha_{\rho}-k) \right],
\end{eqnarray}
whose inverse CDF is denoted as $F_{2M+2k}(x)$. From the Berry-Essen inequality for non identically distributed random variables, the difference between $F_{2M+2k}(x)$ and the standard normal inverse CDF is bounded uniformly, i.e.
\begin{equation}\label{LemmaA_1:sup}
    \sup_{x\in \Re}|F_{2M+2k}(x) - Q(x)|\leq K\phi_1(2M + 2k)\phi_2(2M + 2k),
\end{equation}
where $K$ is a constant, and 
\begin{eqnarray}
     \phi_1(2M + 2k) = &\\& \nonumber \!\!\!\!\!\!\!\!\!\!\!\!\!\!\!\!\!\!\!\!\!\!\!\!\!\!\!\!\!\!\!\!\!\!\!\!\!\!\!\!\!\!\!\!\!\! \bigg( \sum_{m=1}^{2M}\mathbb{E}[|S_m-\mathbb{E}[S_m]|^2] + \sum_{m=1}^{2k}\mathbb{E}[|V_l-\mathbb{E}[V_l]|^2]\bigg)^{-\frac{3}{2}},
\end{eqnarray}
\begin{eqnarray}
   \!\!\!\!\! \phi_2(2M + 2k) = &\\& \nonumber \!\!\!\!\!\!\!\!\!\!\!\!\!\!\!\!\!\!\!\!\!\!\!\!\!\!\!\!\!\!\!\!\!\!\!\!\!\!\!\!\!\!\!\!\!\! \bigg(\sum_{m=1}^{2M}\mathbb{E}[|S_m-\mathbb{E}[S_m]|^3]+\sum_{l=1}^{2k}\mathbb{E}[|V_l-\mathbb{E}[V_l]|^3]\bigg).
\end{eqnarray}
Since $S_m\sim \Gamma(\frac{1}{2},\alpha_{\rho})$ and $V_l\sim \Gamma(\frac{1}{2},1)$, it can be verified by computing the fourth central moment that the second- and third-order central moments are finite. Let $\mathbb{E}[|S_m-\mathbb{E}[S_m]|^3] = \frac{C_1}{2}>0$, $\mathbb{E}[|V_l-\mathbb{E}[V_l]|^3] = \frac{C_2}{2}>0$, $\mathbb{E}[|S_m-\mathbb{E}[S_m]|^2] = \frac{C_3}{2}>0$, and $\mathbb{E}[|V_l-\mathbb{E}[V_l]|^2] = \frac{C_4}{2}>0$. Using these notations, we rewrite (\ref{LemmaA_1:sup}) as follows:
\begin{equation}
    |F_{2M+2k}(x) - Q(x)|\leq \frac{C_{1}M + C_{2}k}{(C_{3}M+C_{4}k)^{\frac{3}{2}}}~~ \forall x.
\end{equation}
Since $C_{1}M + C_{2}k < \max\{C_1,C_2\}(M + k)$ and $C_{3}M + C_{4}k > \min\{C_3,C_4\}(M + k)$ we obtain:
\begin{eqnarray}
    |F_{2M+2k}(x) - Q(x)|\leq &\nonumber\\& \!\!\!\!\!\!\!\!\!\!\!\!\!\!\! \frac{\max\{C_1,C_2\}(M+k)}{\min\{C_3,C_4\}(M+k)^{\frac{3}{2}}} = \frac{C}{\sqrt{M+k}}~~ \forall x,
\end{eqnarray}
where $C = \frac{\max\{C_1,C_2\}}{\min\{C_3,C_4\}}$. This implies: 
\begin{equation*}
    F_{2M+2k}(x) - Q(x) ~=~ \mathcal{O}\left(\frac{1}{\sqrt{M+k}}\right) ~~~\forall x.
\end{equation*}
This can be used to show that:
\begin{eqnarray}
 \Pr\left\{\sum_{m=1}^{2M}S_{m}+\sum_{l=1}^{2k}V_{l} \geq \beta\right\} =
 &\nonumber\\& \!\!\!\!\!\!\!\!\!\!\!\!\!\!\!\!\!\!\!\!\!\!\!\!\!\!\!\!\!\!\!\!\!\!\!\!\!\!\!\!\!\!\!\!\!\!\!\!\!\!\!\!\!\!\!\!\!\!\!\!\!\!\!\!\! F_{2M+2k}(w(k)) ~=~ Q(w(k)) + \mathcal{O}\left(\frac{1}{\sqrt{M+k}}\right). 
\end{eqnarray}
\end{proof}

%%%%%%%%%%%%%%%%%%%%%%%%%%%%%%%%Lemmma 2%%%%%%%%%%%%%%
\begin{lemma}\label{appendix_A:Q_function}
The probability of error in (\ref{Explicit Total Probability}) is given by,
\begin{equation}
p_{e} ~=~ 1 ~-~ \frac{1}{\sqrt{2\pi}}\int_{-\infty}^{\infty}e^{-\frac{s^2}{2}}Q(w(s))ds ~+~ \mathcal{O}\left(\frac{1}{\sqrt{N}}\right).
\end{equation}
\end{lemma}
\begin{proof}
We begin by noting that the expression in (\ref{Explicit Total Probability}) is the expected value of (\ref{large_sum_control}) with respect to a Binomial distribution with parameters $N-1$ and $\tau$. Here, we show that this expected value can be instead taken with respect to a Gaussian distribution with $\mathcal{O}\left(\frac{1}{\sqrt{N}}\right)$ error term. To start with, since the Binomial distribution does not admit a density function it is more convenient to use Stieltjes integrals over a compact interval before going to infinity. First, we show that $\mathcal{O}\left(\frac{1}{\sqrt{M+k}}\right) \,=\, \mathcal{O}\left(\frac{1}{\sqrt{N}}\right)$ by noticing that:
\begin{equation}
    \left|\mathcal{O}\left(\frac{1}{\sqrt{M+k}}\right)\right| < \frac{C}{\sqrt{M+k}} \,=\, \frac{C}{\sqrt{N(\zeta + \delta)}} \,=\, \frac{\widetilde{C}(\delta)}{\sqrt{N}},  
\end{equation}
where $0<\delta = \frac{k}{N}<1$.
We also write the error resulting from integrating with respect to different CDFs as follows:
\begin{eqnarray}\label{Lemma_2:difference}
    \bigg|\int_{-x}^{x}\left(Q(w(k))+\mathcal{O}\left(\frac{1}{\sqrt{N}}\right)\right)dF(k)~~- &\nonumber\\& \!\!\!\!\!\!\!\!\!\!\!\!\!\!\!\!\!\!\!\!\!\!\!\!\!\!\!\!\!\!\!\!\!\!\!\!\!\!\!\!\!\!\!\!\!\!\!\!\!\!\!\!\!\!\!\!\!\!\!\!\!\!\!\!\!\!\!\!\!\!\!  \mathlarger{\int_{-x}^{x}}\left(Q(w(k))+\mathcal{O}\left(\frac{1}{\sqrt{N}}\right)\right)d\widetilde{F}(k) \bigg|,
\end{eqnarray}
where $F(k)$ is the CDF of a Binomial RV with mean $(N-1)\tau$ and variance $(N-1)\tau(1-\tau)$ and $\widetilde{F}(k)$ is the CDF of the approximating Gaussian distribution with the same mean and variance. Since both $F(k)$ and $\widetilde{F}(k)$ are non-decreasing on any compact interval $[-x,x]$, the difference $F(k)-\widetilde{F}(k)$ is of bounded variation, and hence the Stieltjes integral with respect to $F(k)-\widetilde{F}(k)$ is defined for any continuous function. Since $Q(w(k))$ is continuous at $w(k)$, we can re-write (\ref{Lemma_2:difference}) as follows:
\begin{equation}\label{Lemma_2:Stiltjes_crap}
    \left|\int_{-x}^{x}\left(Q(w(k))+\mathcal{O}\left(\frac{1}{\sqrt{N}}\right)\right)d\left(F(k)-\widetilde{F}(k)\right)\right|.
\end{equation}
Furthermore, since $Q(w(k))$ is of bounded variation on $[-x,x]$, the integral in (\ref{Lemma_2:Stiltjes_crap}) can be integrated by parts to yield:
%\begin{eqnarray}
%\left|\int_{-x}^{x}\left(Q(w(k))+\mathcal{O}\left(\frac{1}{\sqrt{N}}\right)\right)d\left(F(k)-\widetilde{F}(k)\right)\right| < \\\ \left|\int_{-x}^{x}\left(1+\mathcal{O}\left(\frac{1}{\sqrt{N}}\right)\right)d\left(F(k)-\widetilde{F}(k)\right)\right|
%\end{eqnarray}
\begin{eqnarray}\label{Lemma_2:More_Stiltjes_crap}
\left|\int_{-x}^{x}\left(Q(w(k))+\mathcal{O}\left(\frac{1}{\sqrt{N}}\right)\right)d\left(F(k)-\widetilde{F}(k)\right)\right| ~=
&\nonumber\\&\!\!\!\!\!\!\!\!\!\!\!\!\!\!\!\!\!\!\!\!\!\!\!\!\!\!\!\!\!\!\!\!\!\!\!\!\!\!\!\!\!\!\!\!\!\!\!\!\!\!\!\!\!\!\!\!\!\!\!\!\!\!\!\!\!\!\!\!\!\!\!\!\!\!\!\!\!\!\!\!\!\!\!\!\!\!\!\!\!\!\!\!\!\!\!\!\!\!\!\!\!\!\!\!\!\!\!\!\!\!\!\!\!\!\!\!\!\!\!\!\!\!\!\!\!\!\! |M_1(x) - M_2(x) - M_3(x)|,
\end{eqnarray}
with
\begin{align}
M_1(x) &= \left(Q(x)+\mathcal{O}\left(\frac{1}{\sqrt{N}}\right)\right)\left(F(x)-\widetilde{F}(-x)\right) \\\ 
    M_2(x) &=\left(Q(-x)+\mathcal{O}\left(\frac{1}{\sqrt{N}}\right)\right)\left(F(-x)-\widetilde{F}(-x)\right)  \\\ M_3(x)&=\int_{-x}^{x}\left(F(k)-\widetilde{F}(k)\right)dQ(w(k)).
\end{align}
%\begin{equation}\label{Lemma2:sum_before_triangle}
  %  \left|F(x)-F(-x)-\widetilde{F}(x)+\widetilde{F}(-x) + \mathcal{O}\left(\frac{1}{\sqrt{N}}\right)\right|
%\end{equation}
Now, by the triangle inequality, it follows from (\ref{Lemma_2:More_Stiltjes_crap}) that: 
\begin{eqnarray}\label{lemma2_further_bound}
\left|\int_{-x}^{x}\left(Q(w(k))+\mathcal{O}\left(\frac{1}{\sqrt{N}}\right)\right)d\left(F(k)-\widetilde{F}(k)\right)\right| ~~\leq
&\nonumber\\&\!\!\!\!\!\!\!\!\!\!\!\!\!\!\!\!\!\!\!\!\!\!\!\!\!\!\!\!\!\!\!\!\!\!\!\!\!\!\!\!\!\!\!\!\!\!\!\!\!\!\!\!\!\!\!\!\!\!\!\!\!\!\!\!\!\!\!\!\!\!\!\!\!\!\!\!\!\!\!\!\!\!\!\!\!\!\!\!\!\!\!\!\!\!\!\!\!\!\!\!\!\!\!\!\!\!\!\!\!\!\!\!\!\!\!\!\!\!\!\!\!\!\!\!\!\!\! |M_1(x)| + |M_2(x)| + |M_3(x)|.
\end{eqnarray}
%(\ref{Lemma2:sum_before_triangle}) is less then 
%\begin{equation}\label{Lemma2:triangle}
 %   \left|F(x)-\widetilde{F}(x)\right|+ \left|\widetilde{F}(-x) - F(-x)\right| + \mathcal{O}\left(\frac{1}{\sqrt{N}}\right)
%\end{equation}
Then, since $0<Q(x)<1$, (\ref{lemma2_further_bound}) can be further bounded by:
\begin{eqnarray}\label{Lemma2:triangle}
    \!\!\!\!\!\!\!\!\!\!\!\!\!\!\!\!\!\!\!\!\!\!K_1\left|F(x)-\widetilde{F}(x)\right|+ K_2\left|\widetilde{F}(-x) - F(-x)\right| ~+ &\nonumber\\&\!\!\!\!\!\!\!\!\!\!\!\!\!\!\!\!\!\!\!\!\!\!\!\!\!\!\!\!\!\!\!\!\!\!\!\!\!\!\!\!\!\!\!\!\!\!\!\!\!\!\!\!\!\!\!\!\!\!\!\!\!\!\!\!\!\!\!\!\!\!\!\!\!\!\!\!\!\!\!\!\!\!\!\!\!\!\!\!\!\!\!\!\!\!\!\sup_{k\in[-x,x]}\left|F(k)-\widetilde{F}(k)\right|V_{-x}^{x}(Q),
\end{eqnarray}
where $V_{-x}^{x}(Q) < \infty$ is the total variation of the Q-function over $[-x,x]$.
Recall, also that $F(x)$ represents the CDF of a large sum of $N-1$ binary random variables. Consequently, upon appropriate normalization and by applying the Berry-Essen inequality for i.i.d random variables we see that (\ref{Lemma2:triangle}) is less than:
\begin{equation}
    \frac{A_1}{\sqrt{N}}+\frac{A_2}{\sqrt{N}} + \frac{A_3}{\sqrt{N}}  ~=~ \mathcal{O}\left(\frac{1}{\sqrt{N}}\right),
\end{equation}
for some positive constants $A_1$, $A_2$ and $A_3$. Now, coming back to (\ref{Lemma_2:difference}) we see that:
\begin{eqnarray}\label{Lemma2:approximation_Pre_final}
    \int_{-x}^{x}\left(Q(w(k))+\mathcal{O}\left(\frac{1}{\sqrt{N}}\right)\right)dF(k) ~=~ &\nonumber\\&\!\!\!\!\!\!\!\!\!\!\!\!\!\!\!\!\!\!\!\!\!\!\!\!\!\!\!\!\!\!\!\!\!\!\!\!\!\!\!\!\!\!\!\!\!\!\!\!\!\!\!\!\!\!\!\!\!\!\!\!\!\!\!\!\!\!\!\!\!\!\!\!\!\!\!\!\!\!\!\!\!\!\!\!\!\!\!\!\!\!\!\!\!\! \mathlarger{\int_{-x}^{x}}\left(Q(w(k))+\mathcal{O}\left(\frac{1}{\sqrt{N}}\right)\right)d\widetilde{F}(k) + \mathcal{O}\left(\frac{1}{\sqrt{N}}\right),
\end{eqnarray}
where the $\mathcal{O}(\frac{1}{\sqrt{N}})$ on the left- and right-hand sides can be safely dropped since both CDFs are bounded, thereby yielding:
\begin{equation}\label{Lemma2:approximation_final}
    \int_{-x}^{x}Q(w(k))dF(k) ~=~ \int_{-x}^{x}Q(w(k))d\widetilde{F}(k) + \mathcal{O}\left(\frac{1}{\sqrt{N}}\right).
\end{equation}
The right-hand side integral in (\ref{Lemma2:approximation_final}) is with respect to a Gaussian CDF and can now be explicitly written in terms of a density by making the substitution $s = \frac{k - (N-1)\tau}{\sqrt{(N-1)\tau(1-\tau)}}$
\begin{equation}
    \int_{-x}^{x}Q(w(k))d\widetilde{F}(k) ~=~ \frac{1}{\sqrt{2\pi}}\int_{-x'}^{x'}Q(w(s))e^{-\frac{s^2}{2}}ds.
\end{equation}
Since the approximation is uniform in $x$ and the integrals are convergent, we can go to the limit in (\ref{Lemma2:approximation_final}) and obtain:
\begin{eqnarray}
    \int_{-\infty}^{\infty}Q(w(k))dF(k) ~~= &\nonumber\\&\!\!\!\!\!\!\!\!\!\!\!\!\!\!\!\!\!\!\!\!\!\!\!\!\!\!\!\!\!\!\!\!\!\!\!
    \frac{1}{\sqrt{2\pi}}\mathlarger{\int_{-\infty}^{\infty}}Q(w(s))e^{-\frac{s^2}{2}}ds + \mathcal{O}\left(\frac{1}{\sqrt{N}}\right).
\end{eqnarray}
By further noticing that the left-hand side of (\ref{Lemma2:approximation_Pre_final}) is nothing but the expression of $1-p_e$, we finally obtain: 
\begin{equation}
    p_e ~=~ 1 - \frac{1}{\sqrt{2\pi}}\int_{-\infty}^{\infty}Q(w(s))e^{-\frac{s^2}{2}}ds + \mathcal{O}\left(\frac{1}{\sqrt{N}}\right).
\end{equation}
\end{proof}

%%%%%%%%%%%%%%%%%%%%%%%%%%%%%%%%%%%%%%%%Lemma 3%%%
\begin{lemma}\label{lemma_3}
$w(s)$ in (\ref{before N-1 Approximation}) can be approximated as follows,
\begin{equation}\label{N-1 approximation}
    w(s) ~=~ \frac{\beta-\alpha_{\rho}M+s\sqrt{N\tau(1-\tau)}+N\tau}{\sqrt{\alpha_{\rho}^{2}M+s\sqrt{N\tau(1-\tau)}+N\tau}}~ +~ \mathcal{O}\left(\frac{1}{\sqrt{N}}\right).
\end{equation}
\end{lemma}
\begin{proof}
We denote the first term in (\ref{N-1 approximation}) by $w_{a}(s)$ and start by re-writing (\ref{before N-1 Approximation}) as:
\begin{multline}\label{lemma3 step 2}
    w(s) ~=~ \frac{\beta-\alpha_{\rho}M+s\sqrt{N\tau(1-\tau)}+N\tau}{\sqrt{\alpha_{\rho}^{2}M+s\sqrt{(N-1)\tau(1-\tau)}+(N-1)\tau}} \\+~ \frac{s\sqrt{\tau(1-\tau)}(\sqrt{N-1}-\sqrt{N})-\tau}{\sqrt{\alpha_{\rho}^{2}M+s\sqrt{(N-1)\tau(1-\tau)}+(N-1)\tau}}.
\end{multline}
Dividing and multiplying the first term in (\ref{lemma3 step 2}) by the denominator of the first term in (\ref{N-1 approximation}) yields:
\begin{multline}\label{lemma3 step 3}
    w(s) = w_{a}(s)\frac{\sqrt{\alpha_{\rho}^{2}M+s\sqrt{N\tau(1-\tau)}+N\tau}}{\sqrt{\alpha_{\rho}^{2}M+s\sqrt{(N-1)\tau(1-\tau)}+(N-1)\tau}} \\+~ \frac{s\sqrt{\tau(1-\tau)}(\sqrt{N-1}-\sqrt{N})-\tau}{\sqrt{\alpha_{\rho}^{2}M+s\sqrt{(N-1)\tau(1-\tau)}+(N-1)\tau}}.
\end{multline}
Recall that $M$ grows proportionally to $N$ with proportionality coefficient $\zeta$ (i.e. $M = \zeta N$). Therefore, after multiplying and dividing the second term in (\ref{lemma3 step 3}) by $\frac{1}{\sqrt{N}}$, one can see it is $\mathcal{O}\left(\frac{1}{\sqrt{N}}\right)$. Using this fact and simplifying the first term, (\ref{lemma3 step 3}) can be written as:
\begin{eqnarray}\label{lemma3 step 4}
     w(s) ~=&\\ \nonumber& \!\!\!\!\!\!\!\!\!\!\!\!\!\!\!\!\!\!\!\!\!\!\!\!w_{a}(s)\bigg(1+\frac{s\sqrt{\tau(1-\tau)}(\sqrt{N}-\sqrt{N-1})+\tau}{\alpha_{\rho}^{2}\zeta N+s\sqrt{(N-1)\tau(1-\tau)}+(N-1)\tau}\bigg)^{\frac{1}{2}} + \mathcal{O}\bigg(\frac{1}{\sqrt{N}}\bigg).
\end{eqnarray}
% \begin{eqnarray}\label{lemma3 step 4}
%     \!\!\!\!\!\!\! w(s) ~=&\nonumber\\& \!\!\!\!w_{a}(s)\bigg(1+\frac{s\sqrt{\tau(1-\tau)}(\sqrt{N}-\sqrt{N-1})+\tau}{\alpha_{\rho}^{2}\zeta N+s\sqrt{(N-1)\tau(1-\tau)}+(N-1)\tau}\bigg)^{\frac{1}{2}} &\nonumber\\&  ~+~ \mathcal{O}\bigg(\frac{1}{\sqrt{N}}\bigg).
% \end{eqnarray}
Taking a first order Taylor expansion of the coefficient of $w_{a}(s)$ in (\ref{lemma3 step 4}) we obtain the desired result.
\end{proof}

%%%%%%%%%%%%%%%%%%%%%%%%%%%%%%%%%%%%%%%%Lemma 4%%%
\begin{lemma}\label{lemma_4}
By neglecting the terms $\frac{\beta}{\sqrt{N}}$ and $\frac{s}{\sqrt{N}}\sqrt{\tau(1-\tau)}$ in (\ref{before neglecting terms}) we obtain: 
\begin{equation}\label{lemma4 eqn}
    w(s) ~=~ \frac{\sqrt{N}(\tau-\alpha_{\rho}\zeta)+s\sqrt{\tau(1-\tau)}}{\sqrt{\alpha_{\rho}^{2}\zeta+\tau}} + \mathcal{O}\left(\frac{1}{\sqrt{N}}\right).
\end{equation}
\end{lemma}
\begin{proof}
We denote the first term in (\ref{lemma4 eqn}) by $w'_{a}(s)$.
Factoring out $\frac{\beta}{\sqrt{N}}$ from the first term in (\ref{before neglecting terms}) and resorting to some simplifications, (\ref{before neglecting terms}) is re-written as follows:
\begin{equation}\label{lemma4 2}
       w(s) ~=~ \frac{\sqrt{N}(\tau-\alpha_{\rho}\zeta)+s\sqrt{\tau(1-\tau)}}{\sqrt{\alpha_{\rho}^{2}\zeta+\frac{s}{\sqrt{N}}\sqrt{\tau(1-\tau)}+\tau}} + \mathcal{O}\left(\frac{1}{\sqrt{N}}\right).
\end{equation}
Then, multiplying and dividing the first term in (\ref{lemma4 2}) by $\sqrt{\alpha_{\rho}^{2}\zeta+\tau}$ leads to:
\begin{equation}\label{lemma4 3}
       w(s) ~=~ w'_{a}(s)\frac{\sqrt{\alpha_{\rho}^{2}\zeta+\tau}}{\sqrt{\alpha_{\rho}^{2}\zeta+\frac{s}{\sqrt{N}}\sqrt{\tau(1-\tau)}+\tau}} + \mathcal{O}\left(\frac{1}{\sqrt{N}}\right),
\end{equation}
which is equivalent to:
\begin{equation}\label{lemma4 4}
       w(s) ~=~ w'_{a}(s)\left(1+\frac{s\sqrt{\tau(1-\tau)}}{(\alpha_{\rho}^{2}\zeta+\tau){\sqrt{N}}}\right)^{-\frac{1}{2}} + \mathcal{O}\left(\frac{1}{\sqrt{N}}\right).
\end{equation}
Taking a first order Taylor expansion of the coefficient of $w'_{a}(s)$ in (\ref{lemma4 4}) we obtain the desired result.
\end{proof}

%%%%%%%%%%%%%%%%%%%%%%%%%%%%%%%%%%%%%%%%Lemma 5%%%
\begin{lemma}\label{lemma_5}
By ignoring the term $s\sqrt{\tau(1-\tau)}$ in (\ref{neglecting stuff with s}), we obtain:
\begin{equation}\label{lemma5 eqn}
    w ~=~ \frac{\sqrt{N}(\tau-\alpha_{\rho}\zeta)}{\sqrt{\alpha_{\rho}^{2}\zeta+\tau}} \,+ \mathcal{O}(1),
\end{equation}
\end{lemma}
\begin{proof}
Since the term $\frac{s\sqrt{\tau(1-\tau)}}{\sqrt{\alpha_{\rho}^{2}\zeta+\tau}}$ is independent of $N$ we make a constant error in neglecting it from (\ref{neglecting stuff with s}).
\end{proof}

%%%%%%%%%%%%%%%%%%%%%%%%%%%%%%%%%%%%%%%%Lemma 6%%%
\begin{lemma}\label{lemma_6}
Taking $Q(w(s))$ outside of the integral in (\ref{approximation of Pi}) leads to 
\begin{equation}
    p_{e} ~=~ 1-Q(w)\,+\,\mathcal{O}\left(\frac{1}{\sqrt{N}}\right).
\end{equation}
\end{lemma}
\begin{proof}
We first rewrite $w(s)$, given in (\ref{neglecting stuff with s}), as follows:
\begin{equation}\label{lemma5 1}
    w(s) ~=~ \frac{\sqrt{N}(\tau-\alpha_{\rho}\zeta)}{\sqrt{\alpha_{\rho}^{2}\zeta+\tau}} + s\sqrt{\frac{\tau(1-\tau)}{\alpha_{\rho}^{2}\zeta+\tau}}+\mathcal{O}\left(\frac{1}{\sqrt{N}}\right),
\end{equation}
and denote the first term by $w$. We denote the error probability in (\ref{approximation of Pi}) by $p_{t}$ and the one in (\ref{taking Q out}) by $p_{a}$. The absolute error between $p_{t}$ and $p_{a}$ is then given by:
\begin{eqnarray}\label{lemma6 1}
    \!\!\!\!\!\!\!\!|p_{t}-p_{a}| ~=~ &\\ \nonumber&\!\!\!\!\!\!\!\!\!\!\!\!\!\!\!\!\!\!\!\!\!\!\!\!\!\!\!\!\!\!\!\!\!\!\!\! \left| \frac{1}{\sqrt{2\pi}}\left(\int_{-\infty}^{\infty}\big(Q(w)-Q(w(s))\big)e^{-\frac{s^2}{2}}ds\right)+\mathcal{O}\left(\frac{1}{\sqrt{N}}\right) \right|.
\end{eqnarray}
Using the definition of the Q-function in (\ref{Q-function definition}) we can rewrite (\ref{lemma6 1}) as follows:
\begin{eqnarray}\label{lemma6 2}
    \!\!\!\!\!\!\!\!\!\!\!\!\!\!\!\!|p_{t}-p_{a}| ~=~ &\\ \nonumber&\!\!\!\!\!\!\!\!\!\!\!\!\!\!\!\!\!\!\!\!\!\!\!\!\!\!\!\!\!\!\!\!\!\!\!\! \left| \frac{1}{2\pi}\left(                     \int_{-\infty}^{\infty}\left(  \int_{w}^{w(s)}e^{\frac{t'^{2}}{2}}dt'\right)e^{-\frac{s^2}{2}}ds\right)+\mathcal{O}\left(\frac{1}{\sqrt{N}}\right) \right|.
\end{eqnarray}
Since $e^{-\frac{t^2}{2}} \leq 1~ \forall t$ we can replace the inner integrand function by 1 thereby leading to:
\begin{align}
    |p_{t}-p_{a}| &\leq \left| \frac{1}{2\pi}\int_{-\infty}^{\infty}(w(s)-w)e^{-\frac{s^2}{2}}ds+\mathcal{O}\left(\frac{1}{\sqrt{N}}\right) \right|.  \nonumber
    \\ 
    \begin{split}\nonumber
        &= \bigg| \frac{1}{2\pi}\int_{-\infty}^{\infty}\left(s\sqrt{\frac{\tau(1-\tau)}{\alpha_{\rho}^{2}\zeta+\tau}} +\mathcal{O}\left(\frac{1}{\sqrt{N}}\right) \right)e^{-\frac{s^2}{2}}ds \\
        &\qquad  +\mathcal{O}\left(\frac{1}{\sqrt{N}}\right) \bigg|.
    \end{split} 
    \\ 
    \begin{split}\label{eqn 88} 
        &= \bigg| \frac{1}{\sqrt{2\pi}}\sqrt{\frac{\tau(1-\tau)}{\alpha_{\rho}^{2}\zeta+\tau}}\int_{-\infty}^{\infty} \frac{s}{\sqrt{2\pi}}e^{-\frac{s^2}{2}}ds \\
        &\qquad +\mathcal{O}\left(\frac{1}{\sqrt{N}}\right) \bigg|. 
    \end{split}
\end{align}
% \begin{eqnarray}\label{lemma6 3}
%     \!\!\!\!\!\!\!\!\!|p_{t}-p_{a}| \leq \left| \frac{1}{2\pi}\int_{-\infty}^{\infty}(w(s)-w)e^{-\frac{s^2}{2}}ds+\mathcal{O}\left(\frac{1}{\sqrt{N}}\right) \right|  \!\!\!\!\!\! &\nonumber\\& \nonumber \!\!\!\!\!\!\!\!\!\!\!\!\!\!\!\!\!\!\!\!\!\!\!\!\!\!\!\!\!\!\!\!\!\!\!\!\!\!\!\!\!\!\!\!\!\!\!\!\!\!\!\!\!\!\!\!\!\!\!\!\!\!\!\!\!\!\!\!\!\!\!\!\!\!\!\!\!\!\!\!\!\!\!\!\!\!\!\!\!\!\!\!\!\!\!\!\!\!\!\!\!\!\!\!\!\!\!\!\!\!\!\!\!\!\!\!\!\!\!\!\!\!\!\!\!\!\!\!\!\!\!\!\!\!\!\!\! = \left| \frac{1}{2\pi}\int_{-\infty}^{\infty}\left(s\sqrt{\frac{\tau(1-\tau)}{\alpha_{\rho}^{2}\zeta+\tau}}+\mathcal{O}\left(\frac{1}{\sqrt{N}}\right) \right)e^{-\frac{s^2}{2}}ds+\mathcal{O}\left(\frac{1}{\sqrt{N}}\right) \right|   \!\!\!\!\!\! &\\& \label{eqn 88} \!\!\!\!\!\!\!\!\!\!\!\!\!\!\!\!\!\!\!\!\!\!\!\!\!\!\!\!\!\!\!\!\!\!\!\!\!\!\!\!\!\!\!\!\!\!\!\!\!\!\!\!\!\!\!\!\!\!\!\!\!\!\!\!\!\!\!\!\!\!\!\!\!\!\!\!\!\!\!\!\!\!\!\!\!\!\!\!\!\!\!\!\!\!\!\!\!\!\!\!\!\!\!\!\!\!\!\!\!\!\!\!\!\!\!\!\!\!\!\!\!\!\!\!\!\!\!\!\!\!\!\!\!\!\!\!\!\!\!\!\!\!\!\!\!\!\!\!\!\!\!\!\!\!\!\!\!\!\!\!\!\!\!\!\!\!\!\!\!\!  = \left| \frac{1}{\sqrt{2\pi}}\sqrt{\frac{\tau(1-\tau)}{\alpha_{\rho}^{2}\zeta+\tau}}\int_{-\infty}^{\infty} \frac{s}{\sqrt{2\pi}}e^{-\frac{s^2}{2}}ds+\mathcal{O}\left(\frac{1}{\sqrt{N}}\right) \right|. \!\!\!\!\!\!
% \end{eqnarray}
The remaining integral in (\ref{eqn 88}) is nothing but the expected value of a zero-mean Gaussian RV and therefore we have:
\begin{equation}
    |p_{t}-p_{a}| \leq \frac{K}{\sqrt{N}},
\end{equation}
for some positive constant $K$. Since the absolute error is $\mathcal{O}\left(\frac{1}{\sqrt{N}}\right)$ we obtain the desired result.
\end{proof}

%%%%%%%%%%%%%%%%%%%%%%%%%%%%%%%%%%%%%%%%Lemma 7%%%
\begin{lemma}\label{lemma_7}
The $\mathcal{O}\left(\frac{1}{\sqrt{N}}\right)$ term in (\ref{Asymptotic Probability of Success}) can be taken out of the denominator in the AoI expression and the AoI expression becomes
\begin{equation}\label{lemma7 eqn}
    \Delta(\zeta,N,\rho,\tau) ~=~ \frac{1}{\tau\left(1-Q\left(\frac{\sqrt{N}(\alpha_{\rho}\zeta-\tau)}{\sqrt{\alpha_{\rho}^2 \zeta +\tau}}\right)\right)} + \mathcal{O}\left(\frac{1}{\sqrt{N}}\right).
\end{equation}
\end{lemma}
\begin{proof}
We denote the first term in (\ref{lemma7 eqn}) by $\Delta_{a}$ and define $\Delta_{t}$ as:
\begin{equation}\label{lemma7 1}
     \Delta_{t} ~\triangleq~ \frac{1}{\tau\left(1-Q\left(\frac{\sqrt{N}(\alpha_{\rho}\zeta-\tau)}{\sqrt{\alpha_{\rho}^2 \zeta +\tau}}\right)+\mathcal{O}\left(\frac{1}{\sqrt{N}}\right) \right)}.
\end{equation}
The absolute error between $\Delta_{t}$ and $\Delta_{a}$ is given by,
\begin{eqnarray}\label{lemma7 2}
    |\Delta_{t}-\Delta_{a}|~~ = &\nonumber\\& \nonumber \!\!\!\!\!\!\!\!\!\!\!\!\!\!\!\!\!\!\!\!\!\!\!\!\!\!\!\!\!\!\!\!\!\!\!\!\! \left|\frac{1}{\tau\left(1-Q\left(\frac{\sqrt{N}(\alpha_{\rho}\zeta-\tau)}{\sqrt{\alpha_{\rho}^2 \zeta +\tau}}\right)+\mathcal{O}\left(\frac{1}{\sqrt{N}}\right) \right)} ~-~ \frac{1}{\tau\left(1-Q\left(\frac{\sqrt{N}(\alpha_{\rho}\zeta-\tau)}{\sqrt{\alpha_{\rho}^2 \zeta +\tau}}\right)\right)} \right|.
\end{eqnarray}
Combining the two terms and resorting to some simplifications we obtain:
\begin{eqnarray}\label{lemma7 3}
    |\Delta_{t}-\Delta_{a}| 
    = \mathcal{O}\left(\frac{1}{\sqrt{N}}\right).
\end{eqnarray}
% \begin{eqnarray}\label{lemma7 3}
%     |\Delta_{t}-\Delta_{a}| ~~ = &\nonumber\\& \nonumber \!\!\!\!\!\!\!\!\!\!\!\!\!\!\!\!\!\!\!\!\!\!\!\!\!\!\!\!\!\!\!\!\!\!\!\!\!\left|\frac{ \mathcal{O}\left(\frac{1}{\sqrt{N}}\right) }{\tau\left(1-Q\left(\frac{\sqrt{N}(\alpha_{\rho}\zeta-\tau)}{\sqrt{\alpha_{\rho}^2 \zeta +\tau}}\right)+\mathcal{O}\left(\frac{1}{\sqrt{N}}\right) \right)\left(1-Q\left(\frac{\sqrt{N}(\alpha_{\rho}\zeta-\tau)}{\sqrt{\alpha_{\rho}^2 \zeta +\tau}}\right)\right)} \right|.\\
%     = \mathcal{O}\left(\frac{1}{\sqrt{N}}\right).
% \end{eqnarray}
The fact that the absolute error between $\Delta_{t}$ and $\Delta_{a}$ is $\mathcal{O}\left(\frac{1}{\sqrt{N}}\right)$ implies the desired result in (\ref{lemma7 eqn}).
\end{proof}

\subsection{Finding the Spectral Efficiency from the Packet Error Probability }\label{Appedix: Theorem 3 conditions}
From (\ref{Asymptotic Probability of Success}), we see that the condition $p_e < \epsilon$ leads to:
\begin{equation}\label{3 lemma 1}
    Q\left(\frac{\sqrt{N}(\alpha_{\rho}\zeta-\tau)}{\sqrt{\alpha_{\rho}^2 \zeta + \tau}}\right) ~<~ \epsilon_{1},
\end{equation}
where $\epsilon_{1} = \epsilon + \mathcal{O}\left(\frac{1}{\sqrt{N}}\right)$. Inverting the Q-function and re-arranging the terms yields:
\begin{equation}\label{3 lemma 2}
    \zeta\left(\zeta-\frac{Q^{-1}(\epsilon_1)^2}{N}\right)\alpha_{\rho}^2 - (2\zeta\tau)\alpha_{\rho} + \tau\left(\tau - \frac{Q^{-1}(\epsilon_1)^2 }{N}\right) ~>~ 0.
\end{equation}
The left-hand side of (\ref{3 lemma 1}) equation is a quadratic form of $\alpha_{\rho}$ whose zeroes are given by: 
\begin{equation}\label{3 lemma 3 1}
    \alpha_{\rho}^{+} ~=~ \tau\frac{1+\sqrt{1 -(1-\frac{Q^{-1}(\epsilon_1)^2}{N\zeta})(1-\frac{Q^{-1}(\epsilon_1)^2}{N\tau})}}{\zeta-\frac{Q^{-1}(\epsilon_1)^2}{N}},
\end{equation}
and
\begin{equation}\label{3 lemma 3 2}
    \alpha_{\rho}^{-} ~=~ \tau\frac{1-\sqrt{1 -(1-\frac{Q^{-1}(\epsilon_1)^2}{N\zeta})(1-\frac{Q^{-1}(\epsilon_1)^2}{N\tau})}}{\zeta-\frac{Q^{-1}(\epsilon_1)^2}{N}},
\end{equation}
Recall that $\alpha_{\rho} = \frac{1}{2^{\rho}-1}$ and since $\rho \geq 0$ then $\alpha_{\rho}$ is nonnegative. 
%%%

\vspace{7mm}
\noindent{\textbf{Case 1:} $\zeta > \frac{Q^{-1}(\epsilon_1)^2}{N}$}
%%%
\vspace{2mm}
\newline
In the case that $\zeta > \frac{Q^{-1}(\epsilon_1)^2}{N}$, the parabola in (\ref{3 lemma 2}) is convex with real roots and, therefore, the inequality in (\ref{3 lemma 2}) is satisfied whenever $\alpha_{\rho}<\alpha_{\rho}^{-}$ or $\alpha_{\rho}>\alpha_{\rho}^{+}$. Note that $\alpha_{\rho}^{+}$ is always positive and any $\alpha_{\rho}>\alpha_{\rho}^{+}$ is a valid solution. However, if $\tau \leq \frac{Q^{-1}(\epsilon_1)^2}{N}$, $\alpha_{\rho}^{-}$ will be non-positive and solutions $\alpha_{\rho}<\alpha_{\rho}^{-}$ are invalid as $\alpha_{\rho}$ must be non-negative. Yet, even in the case $\tau > \frac{Q^{-1}(\epsilon_1)^2}{N}$, $\alpha_{\rho}<\alpha_{\rho}^{-}$ leads to incorrect solutions for small $\epsilon_{1}$. Re-writing (\ref{3 lemma 1}), we have:
\begin{equation}\label{3 lemma 4}
    \frac{\sqrt{N}(\alpha_{\rho}\zeta-\tau)}{\sqrt{\alpha_{\rho}^2 \zeta +\tau}}~>~Q^{-1}(\epsilon_1).
\end{equation}
Now if $\epsilon_1$ is small, i.e. smaller than $\frac{1}{2}$, then $Q^{-1}(\epsilon_1)>0$ and hence the term on the left-hand side of (\ref{3 lemma 4}) should be positive as well. In this case, $\alpha_{\rho}^{-}$ should be greater than $\frac{\tau}{\zeta}$ and re-writing $\alpha_{\rho}^{-}$ as:
\begin{equation}\label{3 lemma 5}
    \alpha_{\rho}^{-} ~=~ \frac{\tau}{\zeta}\left(\frac{1-\sqrt{1 -(1-\frac{Q^{-1}(\epsilon_1)^2}{N\zeta})(1-\frac{Q^{-1}(\epsilon_1)^2}{N\tau})}}{1-\frac{Q^{-1}(\epsilon_1)^2}{\zeta N}}\right),
\end{equation}
we see that this is equivalent to the term inside the brackets in (\ref{3 lemma 5}) being greater than 1. In that case, we have:
\begin{eqnarray}
    1-\sqrt{1 -\left(1-\frac{Q^{-1}(\epsilon_1)^2}{N\zeta}\right)\left(1-\frac{Q^{-1}(\epsilon_1)^2}{N\tau}\right)} &\nonumber\\& \!\!\!\!\!\!\!\!\!\!\!\!\!\!\!\!\!\!\!\!\!\!\!\!\!\!\!\!\!\!\!\!\!\!\!\!\!\!\!\!\!\! >~ 1-\frac{Q^{-1}(\epsilon_1)^2}{\zeta N}.
\end{eqnarray}
This implies:
\begin{eqnarray}
    \left(\frac{Q^{-1}(\epsilon_1)^2}{N \zeta}\right)^2 
    > \left(\frac{Q^{-1}(\epsilon_1)^2}{N \zeta}\right)\left(1-\frac{Q^{-1}(\epsilon_1)^2}{N \tau}\right)&\nonumber\\ \!\!\!\!\!\!\!\!\!\!\!\!\!\!\!\!\!\!\!\!\!\!\!\!\!\!\!\!\!\!\!\!\!\!\!\!+\frac{Q^{-1}(\epsilon_1)^2}{N \tau}.    
    \label{third}
\end{eqnarray}
% \begin{equation}
% \begin{aligned}
%     \left(\frac{Q^{-1}(\epsilon_1)^2}{N \zeta}\right)^2 
%     > 1 -\left(1-\frac{Q^{-1}(\epsilon_1)^2}{N\zeta}\right)\left(1-\frac{Q^{-1}(\epsilon_1)^2}{N\tau}\right)\\  \!\!\!\!\!\!\!\!\!\!\!\!\!\!\!\!\!\!\!\!\!\!\!\!\!\!\!\!\!\!\!\!\!\!\!\!\!\!\!\!\!\!\!\!\!\!\!\!\!\!\!\!\!\!\!\!\!\!\!\!\!\!\!\!\!\!\!\!\!\!\!\!\!\!\!\!\!\!\!\!\!\!\!\!\!\!\!\!\!\!\!\!\!\!\!\!\!\!\!\!\!\!\!\!\!\!\!\!\!\!\!\!\!\!\!\!\!\!\!\!\!\!\!\!\!\!\!\!\!\!\!\!\!\!\!\!\!\!\!\!\!\!\!\!\!\!\!\!\!\!\!\!\!\!\!\!\!\!=\left(\frac{Q^{-1}(\epsilon_1)^2}{N \zeta}\right)\left(1-\frac{Q^{-1}(\epsilon_1)^2}{N \tau}\right)+\frac{Q^{-1}(\epsilon_1)^2}{N \tau}.    
%     \label{third}
%     \end{aligned}
% \end{equation}
Dividing both sides of (\ref{third}) by $\left(\frac{Q^{-1}(\epsilon_1)^2}{N \zeta}\right)$ and rearranging the terms, we obtain:
\begin{equation}\label{3 lemma 7}
    \left(\frac{Q^{-1}(\epsilon_1)^2}{N\zeta}\right)\left(1+\frac{\zeta}{\tau}\right) ~>~ \left(1+\frac{\zeta}{\tau}\right),
\end{equation}
from which it follows that:
\begin{equation}\label{3 lemma 8}
    \zeta ~<~ \frac{Q^{-1}(\epsilon_1)^2}{N}.
\end{equation}
As we are assuming $\zeta > \frac{Q^{-1}(\epsilon_1)^2}{N}$ we have a contradiction and thus $\alpha_{\rho}<\alpha_{\rho}^{-}$ is an invalid solution.
%%%

\vspace{7mm}
\noindent{\textbf{Case 2:} $\zeta < \frac{Q^{-1}(\epsilon_1)^2}{N}$}
%%%
\vspace{2mm}
\newline
In this case, the quadratic form involved in (\ref{3 lemma 2}) is concave. Its zeroes are still given by (\ref{3 lemma 3 1}) and (\ref{3 lemma 3 2}) and solutions to (\ref{3 lemma 2}) are $\alpha_{\rho}^{+}<\alpha_{\rho}<\alpha_{\rho}^{-}$. In this case, $\alpha_{\rho}^{+}$ is always negative and only when $\tau>\frac{Q^{-1}(\epsilon_1)^2}{N}$, $\alpha_{\rho}^{-}$ is positive. Since $\alpha_{\rho}$ must always be non-negative, valid solutions to (\ref{3 lemma 2}) are $0<\alpha_{\rho}<\alpha_{\rho}^{-}$ under the condition that $\tau>\frac{Q^{-1}(\epsilon_1)^2}{N}$. Yet, even we will show that this leads to invalid solutions. If $\epsilon_1$ is small, i.e. less than $\frac{1}{2}$, then $Q^{-1}(\epsilon_1)$ is positive and so $\alpha_{\rho}\zeta-\tau$ must be greater than 0 or we have a contradiction. Assuming that $\alpha_{\rho}^{-}\zeta-\tau$ is positive, it follows that:
\begin{equation}
    \alpha_{\rho}^{-}\zeta-\tau ~>~ 0, 
\end{equation}
which is equivalent to:
% \begin{equation}    
%     \alpha_{\rho}^{-} ~>~ \frac{\tau}{\zeta}
% \end{equation}
\begin{equation}\label{3 lemma 9}
    \frac{\tau}{\zeta}\left(\frac{\sqrt{1+(\frac{Q^{-1}(\epsilon_1)^2}{N \zeta}-1)(1-\frac{Q^{-1}(\epsilon_1)^2}{N \tau})}-1}{\frac{Q^{-1}(\epsilon_1)^2}{N \zeta}-1}\right) ~>~ \frac{\tau}{\zeta},
\end{equation}
where in (\ref{3 lemma 9}) we multiplied the numerator and denominator of $\alpha_{\rho}^{-}$ given in (\ref{3 lemma 3 2}) by $\frac{1}{\zeta}$. Simplifying (\ref{3 lemma 9}) further leads to:
% \begin{align}
%     \sqrt{1+\left(\frac{Q^{-1}(\epsilon_1)^2}{N \zeta}-1\right)\left(1-\frac{Q^{-1}(\epsilon_1)^2}{N \tau}\right)}-1 & ~>~ \frac{Q^{-1}(\epsilon_1)^2}{N \zeta}-1 ,\nonumber \\
%     1+\left(\frac{Q^{-1}(\epsilon_1)^2}{N \zeta}-1\right)\left(1-\frac{Q^{-1}(\epsilon_1)^2}{N \tau}\right) & ~>~ \left(\frac{Q^{-1}(\epsilon_1)^2}{N \zeta}\right)^2 , \nonumber \\ \label{3 lemma 10}
%     \frac{Q^{-1}(\epsilon_1)^2}{N \zeta}+\frac{Q^{-1}(\epsilon_1)^2}{N \tau} - \frac{Q^{-1}(\epsilon_1)^4}{N^2 \zeta \tau} & ~>~ \left(\frac{Q^{-1}(\epsilon_1)^2}{N \zeta}\right)^2 \nonumber ,\\
%     1+\frac{\zeta}{\tau} - \frac{Q^{-1}(\epsilon_1)^2}{N \tau} & ~>~ \frac{Q^{-1}(\epsilon_1)^2}{N \zeta}, \\
%     1+\frac{\zeta}{\tau} &~>~ \frac{Q^{-1}(\epsilon_1)^2}{N}\left(\frac{1}{\zeta}+\frac{1}{\tau}\right) \nonumber ,\\
%     \label{3 lemma 11}
%     \zeta & ~>~ \frac{Q^{-1}(\epsilon_1)^2}{N},
% \end{align}
\begin{equation}
    \zeta  ~>~ \frac{Q^{-1}(\epsilon_1)^2}{N}.
\end{equation}
% where in (\ref{3 lemma 10}) we divided both sides by $\frac{Q^{-1}(\epsilon_1)^2}{N \zeta}$. 
As we are assuming that $\zeta<\frac{Q^{-1}(\epsilon_1)^2}{N}$, we arrive at a contradiction. Therefore, if $\zeta<\frac{Q^{-1}(\epsilon_1)^2}{N}$ there are no valid solutions.

\bibliographystyle{IEEEtran}
\bibliography{IEEEabrv,references}

\begin{thebibliography}{10}
\providecommand{\url}[1]{#1}
\csname url@samestyle\endcsname
\providecommand{\newblock}{\relax}
\providecommand{\bibinfo}[2]{#2}
\providecommand{\BIBentrySTDinterwordspacing}{\spaceskip=0pt\relax}
\providecommand{\BIBentryALTinterwordstretchfactor}{4}
\providecommand{\BIBentryALTinterwordspacing}{\spaceskip=\fontdimen2\font plus
\BIBentryALTinterwordstretchfactor\fontdimen3\font minus
  \fontdimen4\font\relax}
\providecommand{\BIBforeignlanguage}[2]{{%
\expandafter\ifx\csname l@#1\endcsname\relax
\typeout{** WARNING: IEEEtran.bst: No hyphenation pattern has been}%
\typeout{** loaded for the language `#1'. Using the pattern for}%
\typeout{** the default language instead.}%
\else
\language=\csname l@#1\endcsname
\fi
#2}}
\providecommand{\BIBdecl}{\relax}
\BIBdecl

\bibitem{kaul2011minimizing}
S.~Kaul, M.~Gruteser, V.~Rai, and J.~Kenney, ``Minimizing age of information in
  vehicular networks,'' in \emph{2011 8th Annual IEEE Communications Society
  Conference on Sensor, Mesh and Ad Hoc Communications and Networks}, 2011, pp.
  350--358.

\bibitem{liu2018massiveII}
L.~Liu and W.~Yu, ``Massive connectivity with massive mimo—part ii:
  Achievable rate characterization,'' \emph{IEEE Transactions on Signal
  Processing}, vol.~66, no.~11, pp. 2947--2959, 2018.

\bibitem{dutkiewicz2017massive}
E.~Dutkiewicz, X.~Costa-Perez, I.~Z. Kovacs, and M.~Mueck, ``Massive
  machine-type communications,'' \emph{IEEE Network}, vol.~31, no.~6, pp. 6--7,
  2017.

\bibitem{yuan2016non}
Y.~Yuan, Z.~Yuan, G.~Yu, C.-h. Hwang, P.-k. Liao, A.~Li, and K.~Takeda,
  ``Non-orthogonal transmission technology in lte evolution,'' \emph{IEEE
  Communications Magazine}, vol.~54, no.~7, pp. 68--74, 2016.

\bibitem{Wei_Yu_paper}
L.~Liu, E.~G. Larsson, W.~Yu, P.~Popovski, C.~Stefanovic, and E.~De~Carvalho,
  ``Sparse signal processing for grant-free massive connectivity: {A} future
  paradigm for random access protocols in the internet of things,'' \emph{IEEE
  Signal Processing Magazine}, vol.~35, no.~5, pp. 88--99, 2018.

\bibitem{gallager1985perspective}
R.~Gallager, ``A perspective on multiaccess channels,'' \emph{IEEE Transactions
  on information Theory}, vol.~31, no.~2, pp. 124--142, 1985.

\bibitem{chen2017capacity}
X.~Chen, T.-Y. Chen, and D.~Guo, ``Capacity of gaussian many-access channels,''
  \emph{IEEE Transactions on Information Theory}, vol.~63, no.~6, pp.
  3516--3539, 2017.

\bibitem{polyanskiy2017perspective}
Y.~Polyanskiy, ``A perspective on massive random-access,'' in \emph{2017 IEEE
  International Symposium on Information Theory (ISIT)}, 2017, pp. 2523--2527.

\bibitem{liu2018massive}
L.~Liu and W.~Yu, ``Massive connectivity with massive mimo—part i: Device
  activity detection and channel estimation,'' \emph{IEEE Transactions on
  Signal Processing}, vol.~66, no.~11, pp. 2933--2946, 2018.

\bibitem{marzetta2010noncooperative}
T.~L. Marzetta, ``Noncooperative cellular wireless with unlimited numbers of
  base station antennas,'' \emph{IEEE Transactions on Wireless Communications},
  vol.~9, no.~11, pp. 3590--3600, 2010.

\bibitem{kaul2012real}
S.~Kaul, R.~Yates, and M.~Gruteser, ``Real-time status: How often should one
  update?'' in \emph{2012 Proceedings IEEE INFOCOM}.\hskip 1em plus 0.5em minus
  0.4em\relax IEEE, 2012, pp. 2731--2735.

\bibitem{jiang2019timely}
Z.~Jiang, B.~Krishnamachari, X.~Zheng, S.~Zhou, and Z.~Niu, ``Timely status
  update in wireless uplinks: Analytical solutions with asymptotic
  optimality,'' \emph{IEEE Internet of Things Journal}, vol.~6, no.~2, pp.
  3885--3898, 2019.

\bibitem{kadota2019scheduling}
I.~Kadota, A.~Sinha, and E.~Modiano, ``Scheduling algorithms for optimizing age
  of information in wireless networks with throughput constraints,''
  \emph{IEEE/ACM Transactions on Networking}, vol.~27, no.~4, pp. 1359--1372,
  2019.

\bibitem{bastopcu2020partial}
M.~Bastopcu and S.~Ulukus, ``Partial updates: Losing information for
  freshness,'' in \emph{2020 IEEE International Symposium on Information Theory
  (ISIT)}.\hskip 1em plus 0.5em minus 0.4em\relax IEEE, 2020, pp. 1800--1805.

\bibitem{baknina2018sening}
A.~Baknina, S.~Ulukus, O.~Oze, J.~Yang, and A.~Yener, ``Sening information
  through status updates,'' in \emph{2018 IEEE International Symposium on
  Information Theory (ISIT)}, 2018, pp. 2271--2275.

\bibitem{chen2020multiuser}
H.~Chen, Q.~Wang, Z.~Dong, and N.~Zhang, ``Multiuser scheduling for minimizing
  age of information in uplink mimo systems,'' in \emph{2020 IEEE/CIC
  International Conference on Communications in China (ICCC)}.\hskip 1em plus
  0.5em minus 0.4em\relax IEEE, 2020, pp. 1162--1167.

\bibitem{zhu2020status}
Z.~Zhu, B.~Yu, and Y.~Cai, ``Status update performance in uplink massive
  mu-mimo short-packet communication systems,'' in \emph{2020 International
  Conference on Wireless Communications and Signal Processing (WCSP)}.\hskip
  1em plus 0.5em minus 0.4em\relax IEEE, 2020, pp. 115--119.

\bibitem{feng2022precoding}
S.~Feng and J.~Yang, ``Precoding and scheduling for aoi minimization in mimo
  broadcast channels,'' \emph{IEEE Transactions on Information Theory}, 2022.

\bibitem{yu2021age}
B.~Yu and Y.~Cai, ``Age of information in grant-free random access with massive
  mimo,'' \emph{IEEE Wireless Communications Letters}, vol.~10, no.~7, pp.
  1429--1433, 2021.

\bibitem{heath2018foundations}
R.~W. Heath~Jr and A.~Lozano, \emph{Foundations of MIMO Communication}.\hskip
  1em plus 0.5em minus 0.4em\relax Cambridge University Press, 2018.

\bibitem{yang2014quasi}
W.~Yang, G.~Durisi, T.~Koch, and Y.~Polyanskiy, ``Quasi-static multiple-antenna
  fading channels at finite blocklength,'' \emph{IEEE Transactions on
  Information Theory}, vol.~60, no.~7, pp. 4232--4265, 2014.

\bibitem{yates2017status}
R.~D. Yates and S.~K. Kaul, ``Status updates over unreliable multiaccess
  channels,'' in \emph{2017 IEEE International Symposium on Information Theory
  (ISIT)}, 2017, pp. 331--335.

\bibitem{li2014throughput}
B.~Li, R.~Li, and A.~Eryilmaz, ``Throughput-optimal scheduling design with
  regular service guarantees in wireless networks,'' \emph{IEEE/ACM
  Transactions on Networking}, vol.~23, no.~5, pp. 1542--1552, 2014.

\bibitem{shah2000performance}
A.~Shah and A.~M. Haimovich, ``Performance analysis of maximal ratio combining
  and comparison with optimum combining for mobile radio communications with
  cochannel interference,'' \emph{IEEE Transactions on Vehicular Technology},
  vol.~49, no.~4, pp. 1454--1463, 2000.

\bibitem{amalladinne2018coupled}
V.~K. Amalladinne, A.~Vem, D.~K. Soma, K.~R. Narayanan, and J.-F. Chamberland,
  ``A coupled compressive sensing scheme for unsourced multiple access,'' in
  \emph{2018 IEEE International Conference on Acoustics, Speech and Signal
  Processing (ICASSP)}, 2018, pp. 6628--6632.

\bibitem{calderbank2020chirrup}
R.~Calderbank and A.~Thompson, ``Chirrup: a practical algorithm for unsourced
  multiple access,'' \emph{Information and Inference: A Journal of the IMA},
  vol.~9, no.~4, pp. 875--897, 2020.

\bibitem{pradhan2019joint}
A.~Pradhan, V.~Amalladinne, A.~Vem, K.~R. Narayanan, and J.-F. Chamberland, ``A
  joint graph based coding scheme for the unsourced random access gaussian
  channel,'' in \emph{2019 IEEE Global Communications Conference
  (GLOBECOM)}.\hskip 1em plus 0.5em minus 0.4em\relax IEEE, 2019, pp. 1--6.

\bibitem{pradhan2020polar}
A.~K. Pradhan, V.~K. Amalladinne, K.~R. Narayanan, and J.-F. Chamberland,
  ``Polar coding and random spreading for unsourced multiple access,'' in
  \emph{ICC 2020-2020 IEEE International Conference on Communications
  (ICC)}.\hskip 1em plus 0.5em minus 0.4em\relax IEEE, 2020, pp. 1--6.

\bibitem{shyianov2020massive}
V.~Shyianov, F.~Bellili, A.~Mezghani, and E.~Hossain, ``Massive unsourced
  random access based on uncoupled compressive sensing: Another blessing of
  massive mimo,'' \emph{IEEE Journal on Selected Areas in Communications},
  vol.~39, no.~3, pp. 820--834, 2020.

\bibitem{fengler2019massive}
A.~Fengler, G.~Caire, P.~Jung, and S.~Haghighatshoar, ``Massive mimo unsourced
  random access,'' \emph{arXiv preprint arXiv:1901.00828}, 2019.

\bibitem{wainwright2019high}
M.~J. Wainwright, \emph{High-dimensional statistics: A non-asymptotic
  viewpoint}.\hskip 1em plus 0.5em minus 0.4em\relax Cambridge University
  Press, 2019, vol.~48.

\bibitem{fengler2021non}
A.~Fengler, S.~Haghighatshoar, P.~Jung, and G.~Caire, ``Non-bayesian activity
  detection, large-scale fading coefficient estimation, and unsourced random
  access with a massive mimo receiver,'' \emph{IEEE Transactions on Information
  Theory}, vol.~67, no.~5, pp. 2925--2951, 2021.

\bibitem{gallager2013stochastic}
R.~G. Gallager, \emph{Stochastic Processes: Theory for Applications}.\hskip 1em
  plus 0.5em minus 0.4em\relax Cambridge University Press, 2013.

\end{thebibliography}

\end{document}